\documentclass[sigconf]{acmart}

\usepackage[utf8]{inputenc}
\usepackage[T1]{fontenc}

\newcommand{\myparatight}[1]{\smallskip\noindent{\bf {#1}:}~}

\usepackage{amsthm}
\usepackage{amsmath}
 
\usepackage{amssymb}
\usepackage{float}
\usepackage{graphics}
\usepackage{graphicx}
\usepackage[skip=0pt]{caption}
\usepackage{algorithm}
\usepackage{algpseudocode}
\usepackage{bm}
\usepackage{color}
\usepackage{multirow}
\usepackage{makecell}
\usepackage{subfig}
\usepackage{soul}

\usepackage{bigstrut,multirow,rotating}
\usepackage{tabto}
\usepackage{url}
\usepackage{booktabs}
\usepackage{enumitem}
\usepackage{bbm}
\usepackage{dsfont}

\usepackage{pifont}

\usepackage{diagbox}
\usepackage{cancel}
\usepackage{tabularx}

\usepackage{hyperref}
\hypersetup{
	colorlinks=true,
	linkcolor=blue,
	filecolor=magenta,      
	urlcolor=cyan,
	pdftitle={Overleaf Example},
	pdfpagemode=FullScreen,
}

\allowdisplaybreaks

\definecolor{PineGreen}{RGB}{0,139,114}
\definecolor{BrickRed}{RGB}{140,55,62}

\newcommand{\cmark}{{\color{PineGreen}\ding{51}}}%
\newcommand{\xmark}{{\color{BrickRed}\ding{55}}}%

\newcommand{\alg}{{\textsf{BALANCE}}~}
\newcommand{\algns}{{\textsf{BALANCE}}}

\newtheorem{assumption}{Assumption}
\newtheorem{thm}{Theorem}
\newtheorem{lem}{Lemma}
\newtheorem*{remark}{Remark}

\DeclareMathOperator{\EX}{\mathbb{E}}

\usepackage{color, colortbl}

\definecolor{greyL}{RGB}{230,248,255}

\definecolor{ffe8e7L}{RGB}{255,255,248}

\newcommand{\hairi}[1]{{\color{green}}}

\algnewcommand\algorithmicforpara{\textbf{for}}
\algnewcommand\algorithmicdoinparallel{\textbf{do in parallel}}
\algdef{S}[FOR]{ForParallel}[1]{\algorithmicforpara\ #1\ \algorithmicdoinparallel}
\DeclareMathOperator*{\argmin}{arg\,min}
\copyrightyear{2024}
\acmYear{2024}
\setcopyright{rightsretained}
\acmConference[CCS '24]{Proceedings of the 2024 ACM SIGSAC Conference on Computer and Communications Security}{October 14--18, 2024}{Salt Lake City, UT, USA} \acmBooktitle{Proceedings of the 2024 ACM SIGSAC Conference on Computer and Communications Security (CCS '24), October 14--18, 2024, Salt Lake City, UT, USA} \acmDOI{10.1145/3658644.3670307} \acmISBN{979-8-4007-0636-3/24/10}

\makeatletter
\gdef\@copyrightpermission{
  \begin{minipage}{0.3\columnwidth}
   \href{https://creativecommons.org/licenses/by/4.0/}{\includegraphics[width=0.90\textwidth]{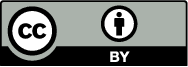}}
  \end{minipage}\hfill
  \begin{minipage}{0.7\columnwidth}
   \href{https://creativecommons.org/licenses/by/4.0/}{This work is licensed under a Creative Commons Attribution International 4.0 License.}
  \end{minipage}
  \vspace{5pt}
}
\makeatother

\begin{document}

\title{Byzantine-Robust Decentralized Federated Learning}

\author{Minghong Fang}
\authornote{Corresponding author}
\affiliation{
	\institution{University of Louisville}
	\country{}
}

\author{Zifan Zhang}
\affiliation{
	\institution{North Carolina State University}
	\country{}
}

\author{Hairi}
\affiliation{
	\institution{University of Wisconsin-Whitewater}
	\country{}
}

\author{Prashant Khanduri}
\affiliation{
	\institution{Wayne State University}
	\country{}
}

\author{Jia Liu}
\affiliation{
	\institution{The Ohio State University}
	\country{}
}

\author{Songtao Lu}
\affiliation{
	\institution{IBM Thomas J. Watson Research Center}
	\country{}
}

\author{Yuchen Liu}
\affiliation{
	\institution{North Carolina State University}
	\country{}
}

\author{Neil Gong}
\affiliation{
	\institution{Duke University}
	\country{}
}

\renewcommand{\shortauthors}{Minghong Fang et al.}

\begin{abstract}
Federated learning (FL) enables multiple clients to collaboratively train machine learning models without revealing their private training data.
In conventional FL, the system follows the server-assisted architecture (server-assisted FL), where the training process is coordinated by a central server.
However, the server-assisted FL framework suffers from poor scalability due to a communication bottleneck at the server, and trust dependency issues.
To address challenges, decentralized federated learning (DFL) architecture has been proposed to allow clients to train models collaboratively in a serverless and peer-to-peer manner.
However, due to its fully decentralized nature, DFL is highly vulnerable to poisoning attacks, where malicious clients could manipulate the system by sending carefully-crafted local models to their neighboring clients.
To date, only a limited number of Byzantine-robust DFL methods have been proposed, most of which are either communication-inefficient or remain vulnerable to advanced poisoning attacks.  
In this paper, we propose a new algorithm called \alg (\underline{B}yzantine-robust \underline{a}veraging through \underline{l}ocal simil\underline{a}rity i\underline{n} de\underline{ce}ntralization) to defend against poisoning attacks in DFL.
In \algns, each client leverages its own local model as a similarity reference to determine if the received model is malicious or benign.
We establish the theoretical convergence guarantee for \alg under poisoning attacks in both strongly convex and non-convex settings. 
Furthermore, the convergence rate of \alg under poisoning attacks matches those of the state-of-the-art counterparts in Byzantine-free settings.
Extensive experiments also demonstrate that \alg outperforms existing DFL methods and effectively defends against poisoning attacks.

\end{abstract}



\maketitle


\section{Introduction} \label{sec:intro}
Federated learning (FL)~\cite{McMahan17} has recently emerged as a powerful distributed learning paradigm that leverages multiple clients to train machine learning models collaboratively without sharing their raw training data.
FL naturally follows the server-based distributed architecture, also known as {\em server-assisted federated learning (server-assisted FL)}~\cite{Konen16,McMahan17}, where the training process is orchestrated by a server. 
However, despite its simplicity, the server-assisted FL framework suffers from {\em three key limitations} due to the reliance of 
a central server.
The first limitation is that the server is vulnerable to the {\em single-point-of-failure risk}, which renders the server a clear target for cyber-attacks or the server itself could experience crashes or other system failures~\cite{dai2022dispfl,guo2021byzantine,he2022byzantine}.
The second limitation of the server-assisted FL framework is that its single-level tree topology implies a {\em communication bottleneck} at the server (the root node) as the number of clients increases~\cite{dai2022dispfl,lian2017can,pasquini2022privacy}.
This communication bottleneck significantly worsens the scalability of large-scale distributed training over a server-based architecture.
The third limitation is that server-assisted FL suffers from {\em trust dependency issues}: all participating clients have to trust the server, which has the potential to influence clients' models arbitrarily~\cite{fowl2021robbing,pasquini2022eluding,pasquini2022privacy}.

\begin{figure}[!t]
	\centering
   	\includegraphics[scale = 0.38]{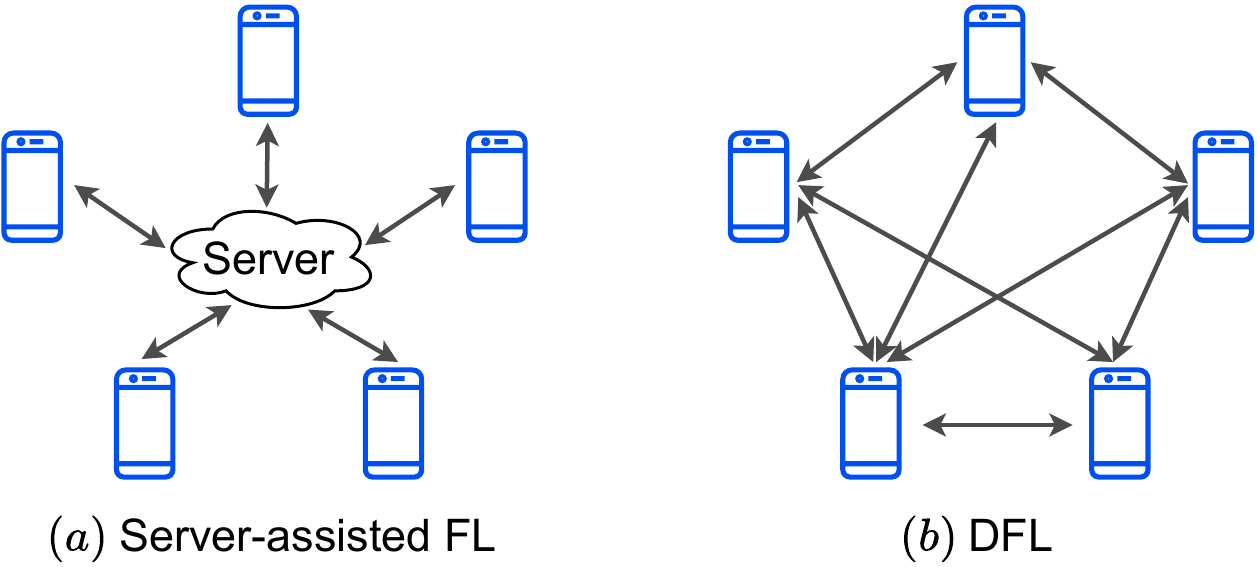}
	\caption{Server-assisted FL vs. DFL.}\label{centralized_decentralized_FL}
\end{figure}

The limitations of these existing server-assisted FL systems have motivated researchers to pursue a fully-decentralized FL design, also known as {\em decentralized federated learning (DFL)}~\cite{beltran2022decentralized,dai2022dispfl,el2021collaborative,kalra2023decentralized,lalitha2018fully,pasquini2022privacy,zhang2022net,liu2023prometheus}.
In DFL, clients exchange information in a peer-to-peer fashion, without the assistance of a server.
Clients in DFL follow the same process of training their local models as in server-assisted FL. However, in this fully-decentralized setting, each client only needs to send its updated model to its neighboring clients and perform local aggregation of the received models.
Fig.~\ref{centralized_decentralized_FL} shows the difference between server-assisted FL and DFL.
Thanks to its salient features, DFL has found a wide range of applications,
such as healthcare services~\cite{ludefl20,nguyen2022novel,rieke2020future} and autonomous driving~\cite{beltran2022decentralized,nguyen2022deep}.

Although DFL provides many benefits, a significant barrier to its widespread adoption is the susceptibility of DFL models to poisoning attacks.
Malicious clients in FL could arbitrarily manipulate the system via poisoning their local training data (aka \textit{data poisoning attacks})~\cite{biggio2012poisoning,munoz2017towards,tolpegin2020data} or local models (aka \textit{model poisoning attacks})~\cite{bagdasaryan2020backdoor,baruch2019little,fang2020local,shejwalkar2021manipulating,xie2020dba,li2022learning} to degrade the learning performance.
To address this challenge, a number of Byzantine-robust FL algorithms have been proposed with varying degrees of success
(e.g.,~\cite{blanchard2017machine,cao2020fltrust,chu2022securing,fu2019attack,kumari2023baybfed,mozaffari2023every,xie2021crfl,Yin18,zhang2022fldetector}).
However, almost all of these existing defenses are based on the server-assisted FL design.
To defend against poisoning attacks in DFL systems, the straightforward approach is to adapt the existing defenses designed for server-assisted FL to DFL setting. However, our later experiments show that directly applying these server-based defenses to DFL leads to unsatisfactory performance since they are not designed for DFL architecture.

We note, however, that designing Byzantine-robust DFL algorithms is highly non-trivial.
One of the main challenges is that, unlike server-assisted FL where the server maintains a single global model; in DFL, each client not only performs model training, but also acts as a ``parameter server'' to aggregate the received models.
At the end of the training process, each client in DFL holds its own final trained model.
The second challenge is that in DFL, clients connect to each other randomly, and different clients may have varying numbers of malicious neighbors.
Furthermore, in DFL, clients interact exclusively with their topological neighbors.
As a result, each client only has a ``partial view'' of the entire system.
The above challenges make it difficult to guarantee that all benign clients in DFL obtain accurate final models, both theoretically and empirically.
Recently, a few Byzantine-robust DFL methods have been proposed~\cite{el2021collaborative,guo2021byzantine,he2022byzantine}.
However, these DFL methods suffer from the following limitations:
First, some approaches lack communication efficiency.
For instance, in LEARN~\cite{el2021collaborative}, each client needs to exchange information with its neighboring clients multiple times during each training round, resulting in a significant communication overhead. 
Second, certain defenses cannot provide theoretical guarantees that all benign clients will obtain accurate final models through the collaborative learning process. 
Moreover, even when such guarantees are provided, these methods need to assume that each benign client has knowledge of its malicious ratio (fraction of neighbors that are malicious).
Third, as our experimental results will demonstrate, existing DFL methods are inherently vulnerable to poisoning attacks.

\myparatight{Our work}In this paper, we aim to bridge this gap.
We propose a novel method called \alg (\ul{B}yzantine-robust \ul{a}veraging through \ul{l}ocal simil\ul{a}rity i\ul{n} de\ul{ce}ntralization) to defend against poisoning attacks in DFL.
Our proposed \alg method is based on the observation that the attacker could manipulate the directions or magnitudes of local models on malicious clients in order to effectively corrupt the FL system.
In our proposed method, each client uses its own local model as a {\em  similarity reference} to assess whether the model it received is malicious or benign.
The high-level idea of \alg is that if the received model is {\em close} to the client's own model in both direction and magnitude, it is considered benign; otherwise, the received model will be ignored.
We provide theoretical guarantees of \alg under poisoning attacks in both strongly convex and non-convex settings.
Specifically, in the case of a strongly convex population loss, we theoretically prove that for our \alg method, the final model learned by each benign client converges to the neighborhood of the global minimum.
In the non-convex setting, we theoretically demonstrate that the final model of each benign client could converge to a neighborhood of a stationary point.
Additionally, the convergence rates of our proposed method in both strongly convex and non-convex settings align with the optimal convergence rate observed in Byzantine-free strongly convex and non-convex optimizations, respectively.
Notably, our theoretical guarantees are established without relying on the stringent and often unrealistic assumptions commonly made in existing DFL methods.
These include the need for the communication graph to be complete and the requirement for each client to know the percentage of their neighbors who are malicious.

We extensively evaluate our proposed method on 5 datasets from different domains, 9 poisoning attacks (including attacks specifically developed for server-assisted FL and those customized for DFL architectures), 12 communication graphs, along
with 8 state-of-the-art FL baselines. 
Furthermore, we explore various practical settings in DFL, including but not limited to, clients having highly non-independent and identically distributed training data (e.g., each client possessing data from merely three classes), clients employing different robust aggregation rules to combine the received models, clients starting with different initial models, various fractions of edges between malicious and benign clients, and time-varying communication graphs (e.g., clients may disconnect from the protocol due to Internet issues).
We summarize our main contributions in this paper as follows:

\begin{list}{\labelitemi}{\leftmargin=1em \itemindent=-0.08em \itemsep=.2em}

\item 
We propose \algns, a novel approach to defend against poisoning attacks in DFL. 
In contrast to existing DFL defenses, 
our \alg algorithm achieves the same communication complexity as that of the state-of-the-art server-assisted FL algorithms.

\item 
We theoretically establish the convergence rate performance of \alg under poisoning attacks in both strongly convex and non-convex settings.
We note that the convergence rate performance of \alg under strongly convex and non-convex settings match the optimal convergence rates in Byzantine-free strongly convex and non-convex optimizations, respectively.

\item 
Our extensive experiments on different benchmark datasets, various poisoning attacks and practical DFL settings demonstrate and verify the efficacy of our proposed \alg method.

\end{list}


\section{Preliminaries and Related Work} \label{sec:related}

\textbf{Notations:} 
Throughout this paper, matrices and vectors are denoted by boldface letters. 
We use $\left\| \cdot \right\|$ for $\ell_2$-norm.
For any given set $\mathcal{V}$, we use $|\mathcal{V}|$ denote its cardinality.

\subsection{Decentralized Federated Learning (DFL)}

Typically, in federated learning (FL), the training procedure can be formulated as an empirical risk minimization (ERM) problem. The aim is to learn a model $\bm{w}^*$ that minimizes the optimization problem expressed as follows:
\begin{align} 
\label{eqn_ERM}
\bm{w}^* = \argmin_{\bm{w} \in \Theta} F(\bm{w}) = \frac{1}{\left| D \right|} \sum\limits_{\zeta \in  D} f(\bm{w}, \zeta),
\end{align}
where $ \Theta \subset \mathbb{R}^d$ is the parameter space, $d$ corresponds to the dimension of model parameter; $F(\cdot)$ denotes the population risk function; $D$ denotes the entire training dataset; $f(\bm{w}, \zeta)$ represents the empirical loss function, which is computed using weight parameter $\bm{w}$ and a training sample $\zeta$.

In this paper, we try to solve the FL problem in~(\ref{eqn_ERM}) in a fully decentralized manner, without requiring assistance from a centralized server.
Specifically, consider a DFL system with a set of clients $\mathcal{V}$.
We let $|\mathcal{V}|$ denote the number of clients in the system.
The network topology of this DFL system is defined by an undirected and unweighted communication graph $\mathcal{G}=(\mathcal{V}, \mathcal{E})$, where $\mathcal{E}$ denotes the set of edges between clients, and self-loops are not allowed.
Communication is only possible between two clients if there is an edge connecting them. 
Each individual client, denoted as $i \in \mathcal{V}$, has its own private training dataset $D_i$.
Collectively, we denote the joint global training dataset as $D=\bigcup_{i \in \mathcal{V}} D_{i}$.
Every client $i$ maintains a model $\bm{w}_i$ that is based on its local training data and information (e.g., model parameter in this paper) gathered from its neighboring clients. 
Specifically, in each training round $t$, each client conducts the following two steps:

 \myparatight{Step I (Local model training and exchanging)}
    Client $i \in \mathcal{V}$ performs local training to get an intermediate model $\bm{w}_i^{t+\frac{1}{2}}$, and subsequently sends $\bm{w}_i^{t+\frac{1}{2}}$ to its neighboring client $j \in \mathcal{N}_i$, where $\mathcal{N}_i$ is the set of neighbors of client $i$, {\em not} including the client $i$ itself.
    At the same time, client $i$ also receives the intermediate local model $\bm{w}_j^{t+\frac{1}{2}}$ from its neighboring client $j$.
    Lines 5-7 in Algorithm~\ref{training_procedure_DFL} summarize Step I. The local training of clients is shown in Algorithm~\ref{local_training}.

     \myparatight{Step II (Local model aggregation)}
     Upon receiving all intermediate local models from its neighbors,
     each client $i \in \mathcal{V}$ aggregates and then updates its model as follows (Line 9 in Algorithm~\ref{training_procedure_DFL}):
	\begin{align}
	\label{local_model_agg}
	  \bm{w}_i^{t+1} = \alpha  \bm{w}_i^{t+\frac{1}{2}} + (1-\alpha) \text{AGG} \{     \bm{w}_j^{t+\frac{1}{2}}, j \in \mathcal{N}_i \},
	\end{align}
	where  $\alpha$ is a trade-off parameter, a larger $\alpha$ indicates trust more in the client’s own intermediate local model, while a smaller value of $\alpha$ means put more weight on the aggregated neighboring models; $\text{AGG} \{ \bm{w}_j^{t+\frac{1}{2}}, j \in \mathcal{N}_i \}$ denotes that the client $i$ employs a certain aggregation rule, represented by $\text{AGG}$, to combine the local models received from its neighboring clients.
    The aggregation rule $\text{AGG}$ can be FedAvg~\cite{McMahan17} or Median~\cite{Yin18}.

In~\cite{pasquini2022privacy}, clients update their local models as \(\bm{w}_i^{t+1} = \text{AGG} \{ \bm{w}_j^{t+\frac{1}{2}}, j \in \widehat{\mathcal{N}_i} \}\), with \(\widehat{\mathcal{N}_i} = \mathcal{N}_i \cup \{ i\}\). Our experiments later reveal that even in non-adversarial settings, such aggregation leads to high error rates in final models using existing Byzantine-robust methods.

\begin{algorithm}[t]
\small
		\caption{LocalTraining($\bm{w}, D, \eta$).}
		\label{local_training}
		\begin{algorithmic}[1]
         \renewcommand{\algorithmicensure}{\textbf{Output:}}
         \Ensure $\bm{w}$. 
		\For {each local iteration}
             \State Sample a mini-batch of training data from $D$ to compute stochastic gradient $\bm{g}(\bm{w})$.
             \State $\bm{w} \leftarrow  \bm{w} - \eta \bm{g}(\bm{w})$.
		\EndFor
		\end{algorithmic}
\end{algorithm}

\begin{algorithm}[t]
\small
		\caption{Training procedure of DFL.}
		\label{training_procedure_DFL}
		\begin{algorithmic}[1]
  	\renewcommand{\algorithmicrequire}{\textbf{Input:}}
		\renewcommand{\algorithmicensure}{\textbf{Output:}}
         \Require Set of clients $\mathcal{V}$; local training data $D_i$, $i \in \mathcal{V}$; training rounds $T$; learning rate $\eta$; communication graph $\mathcal{G}$; parameter $\alpha$; aggregation rule $\text{AGG}$.
         \Ensure Local models $\bm{w}_i^T$, $i \in \mathcal{V}$. 
		   \State Initialize $\bm{w}_i^0$, $i \in \mathcal{V}$.
            \For {$t = 0, 1, \cdots, T-1$}
			\For {each client $i\in \mathcal{V}$ in parallel}
             \State  // Step I: Local model training and exchanging.
			\State $\bm{w}_i^{t+\frac{1}{2}} =\text{LocalTraining($\bm{w}_i^t, D_i, \eta$)}$.
			\State Send $\bm{w}_i^{t+\frac{1}{2}}$ to all neighboring clients $j\in\mathcal{N}_i$.
			\State Receive $\bm{w}_j^{t+\frac{1}{2}}$ from all neighboring clients $j\in\mathcal{N}_i$.
            \State  // Step II: Local model aggregation.
                \State $ \bm{w}_i^{t+1} = \alpha  \bm{w}_i^{t+\frac{1}{2}} + (1-\alpha) \text{AGG} \{     \bm{w}_j^{t+\frac{1}{2}}, j \in \mathcal{N}_i \}$.
			\EndFor
			\EndFor
		\end{algorithmic}
\end{algorithm}

\subsection{Poisoning Attacks to FL} 

FL is vulnerable to both data poisoning attacks~\cite{biggio2012poisoning,munoz2017towards,tolpegin2020data} and model poisoning attacks~\cite{bagdasaryan2020backdoor,baruch2019little,blanchard2017machine,fang2020local,fung2020limitations,gu2017badnets,shejwalkar2022back,xie2020dba,cao2022mpaf,yin2024poisoning,zhang2024poisoning}.
In data poisoning attacks, malicious clients poison their training data. For instance, in a label flipping attack~\cite{tolpegin2020data}, the attacker flips the labels of local training data in malicious clients while leaving the features unchanged.
Malicious clients can also manipulate their local models directly, which are known as model poisoning attacks~\cite{bagdasaryan2020backdoor,blanchard2017machine,fang2020local,gu2017badnets,shejwalkar2022back}.
Depending on the attacker's objective, model poisoning attacks can be categorized as either untargeted attacks~\cite{blanchard2017machine,fang2020local,shejwalkar2022back} or targeted attacks~\cite{bagdasaryan2020backdoor,gu2017badnets,yar2023backdoor}.
In untargeted attacks, the attacker manipulates the FL system in a way that the final learned model will make incorrect predictions on a significant number of test examples without distinction. 
Conversely, in targeted attacks, the attacker seeks to influence the predictions of the final learned model on the specific inputs.
A recent study~\cite{pasquini2022privacy} demonstrates that DFL is vulnerable to privacy attacks, this topic is out of the scope of this paper.

\begin{table}[!t]
	\small
	 \addtolength{\tabcolsep}{-3.95pt}
	\centering
	\caption{
		``Convex guarantee'' and ``Non-convex guarantee'' mean the method provides theoretical performance guarantees under strongly convex and non-convex settings, respectively; 
		``No know. about $c_i$'' means benign client $i$ has no knowledge about $c_i$ (malicious ratio of client $i$);
		``No complete graph assum.'' means the approach does not require the assumption that the communication graph $\mathcal{G}$ must be a complete graph.;
		``No extra comm. cost'' means the method does not incur extra communication cost compared to FedAvg.
	}
	\label{comp_diff_method}
	\begin{tabular}{|c|c|c|c|c|c|c|}
		\hline
		Method & \multicolumn{1}{c|}{\makecell {Convex \\ guarantee}}  
		 & \multicolumn{1}{c|}{\makecell {Non-convex \\ guarantee}}      
		 & \multicolumn{1}{c|}{\makecell {No know. \\ about $c_i$}} 
		 & \multicolumn{1}{c|}{\makecell {No complete \\ graph assum.}} 
		 & \multicolumn{1}{c|}{\makecell {No extra \\ comm. cost}} \\
		\hline
		UBAR~\cite{guo2021byzantine}   & \xmark  &    \xmark         &  \xmark &  \xmark & \cmark \\
		\hline
		LEARN~\cite{el2021collaborative}     &   \xmark      &    \cmark        & \xmark & \xmark  &  \xmark \\
		\hline
		SCCLIP~\cite{he2022byzantine}    &   \xmark     &  \cmark      & \cmark   & \cmark & \cmark\\
		\hline
		\alg  &  \cmark      &   \cmark         &  \cmark & \cmark & \cmark   \\
		\hline
	\end{tabular}
\end{table}

\subsection{Byzantine-robust DFL Aggregation Rules}

Since FL is vulnerable to poisoning attacks, many Byzantine-robust aggregation mechanisms for FL have been developed~\cite{blanchard2017machine,cao2020fltrust,chu2022securing,fu2019attack,karimireddy2020byzantine,kumari2023baybfed,Mhamdi18,munoz2019byzantine,pan2020justinian,Yin18,cao2022flcert,cao2023fedrecover,fang2022aflguard,xu2024robust}.
However, most of them are based on the server-assisted FL design.
Recently, a few Byzantine-robust FL methods have been proposed to tackle this challenge in DFL setting~\cite{el2021collaborative,guo2021byzantine,he2022byzantine}.
For instance, in the UBAR~\cite{guo2021byzantine} method, each client first selects a group of neighboring clients that has the smallest sum of distances to its own local model, then further excludes information from neighbors that would result in a larger loss.
LEARN~\cite{el2021collaborative} is another type of Byzantine-robust aggregation protocol designed for DFL.
Clients in LEARN share both local model updates and local models with their neighboring clients in each training round. 
Specifically, each client first exchanges local model update with neighboring clients for $\lceil \log_2 t \rceil$ times, and then shares its local model one time, where $t$ is the current training round. The trimmed mean aggregation rule is used by clients to combine the received local model updates and models.
In the SCCLIP~\cite{he2022byzantine} aggregation rule, each client clips all received local models from neighboring clients to make sure the norm of a clipped received local model is no larger than that of the client's own model.

However, there are several inherent limitations in existing DFL defenses.
First, UBAR cannot theoretically guarantee that every benign client could learn an accurate model.
Second, although some methods offer theoretical guarantees for benign clients, they assume that each benign client $i$ has knowledge of its malicious ratio $c_i$, 
which is computed as number of malicious neighbors divided by the total number of neighbors of client $i$.
That is to say, in these methods, it is assumed that each benign client knows the number of neighbors that are malicious.
Third, the LEARN method additionally presupposes that the underlying communication graph \(\mathcal{G}\) must be a complete (fully connected) graph.
Notably, the LEARN method incurs a large communication cost, as during each training round, clients need to exchange information with their peers several times.
Contrary to existing DFL defenses, our proposed \alg method addresses the above limitations.
We compare our proposed \alg with existing Byzantine-robust DFL methods, and summarize the comparison in Table~\ref{comp_diff_method}.
It is important to note that in Table~\ref{comp_diff_method}, we do not compare our method with server-assisted FL methods such as Krum and Median. This is because the theoretical results of server-assisted FL methods cannot be straightforwardly transferred to the DFL framework, owing to significant variations in their architectures and operational procedures. DFL necessitates distinct theoretical frameworks and analyses tailored to its specific features and obstacles. Determining how to adapt the theoretical results of server-assisted FL methods to the DFL context is a challenging task and falls outside the scope of this paper.


\section{Problem Statement} 
\label{problem_statement}

\myparatight{Threat Model}
Similar to prior works~\cite{fang2020local,guo2021byzantine,he2022byzantine,shejwalkar2021manipulating,shejwalkar2022back}, we assume that the attacker controls some malicious clients, those malicious clients could either poison their local training data or directly manipulate the local models that are sent to their neighboring clients.
Note that each malicious client could only send malicious local models to its neighbors.
Additionally, a malicious client could distribute  different local models to different neighboring clients.
We also remark that following~\cite{guo2021byzantine,he2022byzantine, pasquini2022privacy}, the attacker cannot change the communication graph $\mathcal{G}$ between clients. 
However, clients may disconnect from the DFL protocol 
in each training round because of  Internet-related issues.

\myparatight{Attacker’s Knowledge}
In terms of attacker's knowledge, following~\cite{cao2020fltrust,guo2021byzantine,he2022byzantine}, we consider the worst-case attack scenario where the attacker has full-knowledge about the FL system, which includes local training data, the aggregation rule and trade-off parameter $\alpha$ utilized by all clients, as well as the communication graph $\mathcal{G}$.
Note that in both non-adversarial and adversarial scenarios, each client knows the local models of its neighbors since local models are exchanged in DFL.

\myparatight{Defender’s Knowledge and Goal}
The proposed defense has no knowledge about the attacker's strategy nor the communication graph $\mathcal{G}$, but is expected to be capable of withstanding powerful adaptive attacks. 
It is important to note that in the proposed defense, each benign client is unaware of the total number of malicious clients in the system nor the number of neighboring clients that are malicious.
We aim to develop a reliable and resilient DFL approach that meets the following three key goals. 1) Competitive learning performance: the proposed defense scheme for DFL should be effective in non-adversarial settings. Specifically, when there are no malicious clients, the model learned by each benign client using our proposed algorithm should attain comparable test error rate performance to that of averaging-based aggregation, which achieves state-of-the-art performance in non-adversarial settings;
2) Byzantine robustness: the proposed DFL method should demonstrate both theoretical and empirical resilience against Byzantine attacks;
and 3) Communication and computation efficiency: the proposed algorithm should not result in any additional communication or computation costs when compared to FedAvg~\cite{McMahan17} in the absence of attacks.


\section{The \alg Algorithm} \label{sec:defense}

\begin{figure}[!t]
	\centering
	\includegraphics[scale = 0.7]{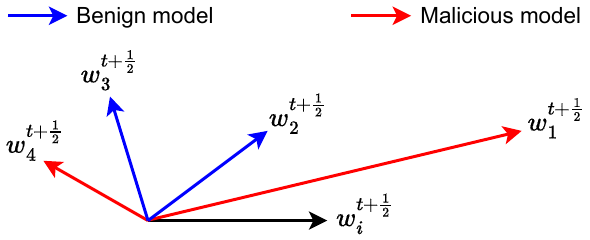}
      \vspace{.01in}
	\caption{Illustration of our proposed \alg method.}
    \label{fig_high_level}
\end{figure}

As summarized in Table~\ref{comp_diff_method}, existing DFL methods 
either would incur a large communication cost, or make strong assumption that each benign client needs to know its malicious ratio.
However, this assumption does not hold in practical scenarios, as in DFL, different clients connect to a different number of malicious clients. 
Moreover, ensuring that every benign client can learn an accurate final model by exchanging information with other clients presents a significant challenge.
In this section, we aim to design a simple yet effective DFL method to achieve three goals defined in Section~\ref{problem_statement}.

As shown in~\cite{fang2020local,shejwalkar2021manipulating}, the attacker could launch model poisoning attacks on FL either by manipulating the directions or magnitudes of the local models on malicious clients. 
In our proposed \algns, if the received intermediate model differs significantly from the client's own intermediate model, it is assumed to be potentially malicious and is ignored.
Specifically, at training round $t$, when client $i \in \mathcal{V}$ receives a intermediate model $\bm{w}_j^{t+\frac{1}{2}}$ from its neighboring client $j \in \mathcal{N}_i$, it uses its own intermediate model $\bm{w}_i^{t+\frac{1}{2}}$ as a {\em similarity reference} to check whether the received model $\bm{w}_j^{t+\frac{1}{2}}$ is malicious or benign.
If $\bm{w}_j^{t+\frac{1}{2}}$ is {\em close} to $\bm{w}_i^{t+\frac{1}{2}}$, both in terms of direction and magnitude, then client $i$ will consider $\bm{w}_j^{t+\frac{1}{2}}$ as a benign model; 
otherwise client $i$ will disregard $\bm{w}_j^{t+\frac{1}{2}}$.
Additionally, as the model approaches convergence, $\bm{w}_j^{t+\frac{1}{2}}$ becomes more closer to $\bm{w}_i^{t+\frac{1}{2}}$. Based on the above insights, client $i$ will only accept $\bm{w}_j^{t+\frac{1}{2}}$ if the following condition holds:
\begin{align}
\label{our_agg}
\| \bm{w}_i^{t+\frac{1}{2}} - \bm{w}_j^{t+\frac{1}{2}} \| \le \gamma \cdot \exp(-\kappa \cdot \lambda(t))   \| \bm{w}_i^{t+\frac{1}{2}}    \|,
\end{align}
where the parameter $\gamma > 0$ sets an upper limit for accepting a model as benign. The value of $\kappa > 0$ determines the rate at which the exponential function decreases; a larger $\kappa$ results in a faster decay, while a smaller $\kappa$ leads to a slower decay. The function $\lambda(t)$ is a monotonically increasing and non-negative function associated with the training round index $t$, meaning it becomes larger as $t$ increases. When $\gamma$ and $\kappa$ are fixed, the term $\gamma \cdot \exp(-\kappa \cdot \lambda(t))$ decreases as $t$ increases.
Various methods exist for choosing \(\lambda(t)\). For instance, a straightforward approach is to define it as \(\lambda(t) = \frac{t}{T}\), where $T$ is the total number of training rounds.

Fig.~\ref{fig_high_level} shows the high-level idea of proposed \alg defense.
In Fig.~\ref{fig_high_level}, $\bm{w}_1^{t+\frac{1}{2}}$, $\bm{w}_2^{t+\frac{1}{2}}$, $\bm{w}_3^{t+\frac{1}{2}}$, and $\bm{w}_4^{t+\frac{1}{2}}$ are four intermediate models sent from neighboring clients of client $i$; $\bm{w}_i^{t+\frac{1}{2}}$ is client $i$'s own intermediate model.
Client $i$ will accept $\bm{w}_2^{t+\frac{1}{2}}$ and $\bm{w}_3^{t+\frac{1}{2}}$ because they are close to $\bm{w}_i^{t+\frac{1}{2}}$.
However, models $\bm{w}_1^{t+\frac{1}{2}}$ and $\bm{w}_4^{t+\frac{1}{2}}$ are flagged as malicious since $\bm{w}_1^{t+\frac{1}{2}}$ deviates significantly from $\bm{w}_i^{t+\frac{1}{2}}$ in terms of magnitude, and $\bm{w}_4^{t+\frac{1}{2}}$ considerably differs from $\bm{w}_i^{t+\frac{1}{2}}$ in terms of direction.

During training round $t$, we define the set $\mathcal{S}_i^t \subseteq \mathcal{N}_i$ as the collection of neighboring clients of client $i$ whose intermediate models satisfy Eq.~(\ref{our_agg}). 
Client $i$ then aggregates the received intermediate models from its neighboring clients by computing the average of all accepted models as $\frac{1}{| \mathcal{S}_i^{t } |} \sum\nolimits_{j \in \mathcal{S}_i^{t}} \bm{w}_j^{t+\frac{1}{2}}$. 
Finally, client $i$ could update its model by combining its own intermediate model with the aggregated intermediate model in the following manner:
\begin{align}
\label{our_agg_combine}
    \bm{w}_i^{t+1} = \alpha  \bm{w}_i^{t+\frac{1}{2}} + (1-\alpha) \frac{1}{| \mathcal{S}_i^{t } |} \sum\nolimits_{j \in \mathcal{S}_i^{t}} \bm{w}_j^{t+\frac{1}{2}}.
\end{align}

Algorithm~\ref{our_alg} shows the pseudocode of our proposed \alg algorithm. 
During the training round $t$, each client executes Lines 4-15 of Algorithm~\ref{our_alg} in parallel.
Specifically, for client $i\in \mathcal{V}$, it first performs local model training and exchanging (Lines 5-7). 
Note that the LocalTraining procedure is shown in Algorithm~\ref{local_training}.
If client $i$ is a malicious client, it can choose to send arbitrary or carefully-crafted intermediate models to its neighboring clients at Line 6.
After that, client $i$ accepts intermediate models that satisfy Eq.~(\ref{our_agg}) and further updates its local model (Lines 9-15).

\begin{algorithm}[t]
\small
		\caption{\algns.}
		\label{our_alg}
		\begin{algorithmic}[1]
  	\renewcommand{\algorithmicrequire}{\textbf{Input:}}
		\renewcommand{\algorithmicensure}{\textbf{Output:}}
         \Require Set of clients $\mathcal{V}$; local training data $D_i$, $i \in \mathcal{V}$; training rounds $T$; learning rate $\eta$; communication graph $\mathcal{G}$; parameters $\alpha$, $\gamma$, $\kappa$ and $\lambda(t)$.
         \Ensure Local models $\bm{w}_i^T$, $i \in \mathcal{V}$. 
		   \State Initialize $\bm{w}_i^0$, $i \in \mathcal{V}$.
            \For {$t = 0, 1, \cdots, T-1$}
			\For {each client $i\in \mathcal{V}$ in parallel}
             \State  // Step I: Local model training and exchanging.
			\State $\bm{w}_i^{t+\frac{1}{2}} =\text{LocalTraining($\bm{w}_i^t, D_i, \eta$)}$.
			\State Send $\bm{w}_i^{t+\frac{1}{2}}$ to all neighboring clients $j\in\mathcal{N}_i$.
			\State Receive $\bm{w}_j^{t+\frac{1}{2}}$ from all neighboring clients $j\in\mathcal{N}_i$.
            \State  // Step II: Local model aggregation.
            \State $\mathcal{S}_i^t=\emptyset$.
            \For {each client $j \in \mathcal{N}_i$ }
             \If {Eq.~(\ref{our_agg}) satisfies}
             \State $\mathcal{S}_i^t=\mathcal{S}_i^t \bigcup \{j \} $.
            \EndIf
		  \EndFor
            \State $\bm{w}_i^{t+1} = \alpha  \bm{w}_i^{t+\frac{1}{2}} + (1-\alpha) \frac{1}{| \mathcal{S}_i^{t } |} \sum\nolimits_{j \in \mathcal{S}_i^{t}} \bm{w}_j^{t+\frac{1}{2}}$.
		  \EndFor
             \EndFor
		\end{algorithmic}
\end{algorithm}

\noindent
{\bf Complexity analysis:} 
In our proposed \alg method, at training round $t$, client $i \in \mathcal{V}$ computes the distance between its own intermediate model $\bm{w}_i^{t+\frac{1}{2}}$ and the received intermediate model $\bm{w}_j^{t+\frac{1}{2}}$ from a neighboring client $j \in \mathcal{N}_i$.
Since the dimension of local model is $d$, and client $i$ has $|\mathcal{N}_i|$ neighboring clients, the computational complexity of each client in our method is $\mathcal{O}(d|\mathcal{N}_i|)$ at each training round.

\section{Theoretical Performance Analysis} \label{sec:analysis}

In this section, we present the convergence performance guarantee of our proposed \algns. 
We let $\mathcal{B} \subseteq \mathcal{V}$ be the set of benign clients.
Let $\mathcal{G}_B$ be the subgraph induced by benign clients.
In a convex setting, we denote the global minimum as $\bm{w}^*$; while in a non-convex setting, $\bm{w}^*$ represents a stationary point (a point which has zero gradient).
Before introducing the theoretical results,
we first present some technical assumptions that are standard in the literature~\cite{el2021collaborative,he2022byzantine,karimireddy2020byzantine,lian2017can,Yin18}.

\begin{assumption}
\label{assumption_1}
The population risk $F(\bm{w})$ is $\mu$-strongly convex, i.e., for all  $\bm{w}_1, \bm{w}_2 \in \Theta$, one has that:
\begin{align}
	F(\bm{w}_1) + \left\langle {\nabla F(\bm{w}_1),\bm{w}_2 - \bm{w}_1} \right\rangle + \frac{\mu}{2}{\left\| {\bm{w}_2 - \bm{w}_1} \right\|^2} \le
	F(\bm{w}_2). \nonumber
\end{align}
\end{assumption}

\begin{assumption}
\label{assumption_2}
The population risk $F(\bm{w})$ is $L$-smooth, i.e., for all  $\bm{w}_1, \bm{w}_2 \in \Theta$, we have that:
\begin{align}	
\left\| {\nabla F(\bm{w}_1) - \nabla F(\bm{w}_2) } \right\| \le L\left\| {\bm{w}_1 - \bm{w}_2 } \right\|. \nonumber
\end{align}
\end{assumption}

\begin{assumption}
\label{assumption_3}
The stochastic gradient $\bm{g}(\bm{w}_{i})$ computed by a benign client $i \in \mathcal{B}$ is an unbiased estimator of the true gradient, and $\bm{g}(\bm{w}_{i})$ has bounded variance, where $\mathcal{B}$ is the set of benign clients.
That is, $\forall i \in \mathcal{B}$, one has that:
\begin{align}
\EX [\bm{g}(\bm{w}_{i}) ] = \nabla F(\bm{w}_{i}), \quad \EX [ \|  \bm{g} (\bm{w}_{i}) -  \nabla F(\bm{w}_{i}) \|]^2 \le  \delta ^2.  
\nonumber
\end{align}
\end{assumption}

\begin{assumption}
\label{assumption_4}
For any benign client $i \in \mathcal{B}$, the model $\bm{w}_{i}$ and $\|  \nabla F(\bm{w}_{i}) \| $ are bounded. 
That is, $\forall i \in \mathcal{B}$, we have $\| \bm{w}_{i} \|  \le \psi$, and $  \| \nabla F(\bm{w}_{i}) \| \le \rho$.
\end{assumption}

\begin{assumption}
\label{assumption_6}
$\mathcal{G}_B$ is connected.
\end{assumption}

With the above assumptions, we provide the theoretical findings of our \alg both in strongly convex and non-convex settings. 
In the strongly convex setting, we have the following result.
\begin{thm}[The Strongly Convex Setting]
\label{Theorem1}
Suppose Assumptions~\ref{assumption_1}-\ref{assumption_6} hold, clients select intermediate models according to Eq.~(\ref{our_agg}).
Let the learning rate $\eta$ and $\gamma$ be chosen as such that $\eta \le \min\{\frac{1}{4L},\frac{1}{\mu}\}$ and $\gamma \le \frac{\rho }{L\psi(1-\alpha)}$.
The value of $\kappa \cdot \lambda(t)$ is larger than 0.
After $T$ training rounds, for any benign client $i \in \mathcal{B}$, it holds that:
\begin{align}
\EX [ F(\bm{w}_{i}^{T})  - F(\bm{w}^*) ] 
& \le (1-\mu \eta)^T  [   F(\bm{w}_{i}^0)  - F(\bm{w}^*) ]\nonumber \\
& \quad + \frac{ 2 L \eta \delta^2}{\mu}  + \frac{2 \gamma   \rho \psi (1-\alpha)}{\mu \eta}, \nonumber
\end{align}
where $\bm{w}_{i}^0$ is client $i$'s initial model.
 \end{thm}	
\begin{proof}
The proof is relegated to Appendix~\ref{sec:appendix_1}.
\end{proof}

Theorem~\ref{Theorem1} says  that by choosing an appropriate learning rate $\eta$ and $\gamma$, for any benign client, the final learned model converges to the neighborhood of the global minimum. 
More importantly, the linear convergence rate is the {\em same} as the optimal convergence rate in Byzantine-free strongly convex optimization algorithm. 
In other words, Byzantine attacks do not hurt the convergence rate of our proposed method.
Note that in this paper, we only consider standard gradient descent, without using the higher-order information or relying on accelerated methods.

In the non-convex setting, we have the following result.
\begin{thm}[The Non-convex Setting]
\label{Theorem2}
Under Assumptions~\ref{assumption_2}-\ref{assumption_6}, clients select intermediate models according to Eq.~(\ref{our_agg}).
Select  $\eta$ and $\gamma$ such that $\eta \le \min\{\frac{1}{4L},\frac{1}{\mu}\}$ and $\gamma \le \frac{\rho }{L\psi(1-\alpha)}$.
In addition, $\kappa \cdot \lambda(t)>0$.
After $T$ training rounds, the following holds for any benign client $i \in \mathcal{B}$:
\begin{align}
\frac{1}{T}  \sum\limits_{t=0}^{T-1}   \EX [   \| \nabla F(\bm{w}_{i}^t) \|^2 ] 
 \le \frac{2  [   F(\bm{w}_{i}^0)  -  F(\bm{w}^*) ]}{\eta T}   + 4 L \eta \delta^2  
 + \frac{4 \gamma   \rho\psi (1-\alpha)}{\eta}. \nonumber
\end{align}
 \end{thm}	
\begin{proof}
The proof is relegated to Appendix~\ref{sec:appendix_2}.
\end{proof}

Theorem~\ref{Theorem2} shows that by 
selecting suitable parameters, the final model of each benign client could converge to the neighborhood of a stationary point. 
The sub-linear convergence rate aligns with the optimal convergence rate in a Byzantine-free non-convex optimization algorithm. 
In other words, poisoning attacks do not impact the convergence rate of our \alg in the non-convex setting.

\begin{remark} 
Assumption~\ref{assumption_3} states that the training data among clients are independent and identically distributed (IID), but this is not required for our experiments.
In Assumption~\ref{assumption_6}, it is posited that the subgraph $\mathcal{G}_B$, which is formed by benign clients, remains connected after eliminating all malicious clients and their corresponding edges. This assumption is critical, as it prevents the scenario where a benign client is exclusively surrounded by malicious neighbors.
Our \alg does not require benign clients to be aware of the architecture of communication graph $\mathcal{G}$. Additionally, in our theoretical analysis, there is no necessity for $\mathcal{G}$ to have a particular architecture, such as the requirement for it to be a complete graph, as assumed in~\cite{el2021collaborative}.
\end{remark}

\begin{remark} 
{
For a strongly-convex (or convex) objective, our Theorem~\ref{Theorem1} guarantees the convergence to a global optimal solution. For a non-convex objective, since guaranteeing
convergence to a global optimal is NP-Hard~\cite{nesterov2018lectures}, guaranteeing convergence to a stationary point (local optimal) is the best one can hope for in the non-convex case.
The convergence rates of our \alg method in both strongly convex and non-convex scenarios match the  best-known convergence rates of their Byzantine-free counterparts.
We also note that using our proposed \alg does not require knowing the precise values of certain parameters introduced in assumptions, such as \( \delta\), \( \psi\), and \( \rho\). Clients only need to check whether Eq.~(\ref{our_agg}) is met while filtering out malicious local models.
Since \(\gamma \cdot \exp(-\kappa \cdot \lambda(t))\) is always upper bounded by \(\gamma\), Theorem~\ref{Theorem1} and Theorem~\ref{Theorem2} rely on the condition that the value of \(\gamma\) is bounded.
}
\end{remark}


\section{Experiments} \label{sec:exp}

\subsection{Experimental Setup}

\subsubsection{Datasets and Poisoning Attacks}

In our experiment, we assess our method and various baselines across multiple datasets, including a synthetic dataset and four real-world datasets: MNIST~\cite{lecun2010mnist}, Fashion-MNIST~\cite{xiao2017online}, Human Activity Recognition (HAR)~\cite{anguita2013public}, and CelebA~\cite{liu2015faceattributes}. Notably, HAR, sourced from 30 smartphone users (each representing a client), exemplifies a real-world FL dataset. Details on the creation of the synthetic dataset and specifics of the other four datasets are available in Appendix~\ref{sec:datasets_app}.

We first consider seven poisoning attacks, including two data poisoning attacks (Label flipping (LF) attack~\cite{tolpegin2020data}, Feature attack) and five model poisoning attacks (Gaussian (Gauss) attack~\cite{blanchard2017machine}, Krum attack~\cite{fang2020local}, Trim attack~\cite{fang2020local}, Backdoor attack~\cite{bagdasaryan2020backdoor,gu2017badnets}, and Adaptive (Adapt) attack~\cite{shejwalkar2021manipulating}).
Note that Backdoor attack is a targeted attack model, where the attacker aims to craft the system such that the final trained model makes incorrect predictions on inputs selected by the attacker.
Adapt attack is the  most powerful attack, where the attacker has full knowledge of the system, including all benign clients' local models and the proposed aggregation rule \alg used by clients.
The attacker in Adapt attack introduces minor perturbation to the benign local models to create malicious models.
The detailed description of seven poisoning attacks is shown in Appendix~\ref{sec:attack_app}.
Additionally, we evaluate two other attacks in Section~\ref{sec:discussion_limitation}: ``a little is enough'' (LIE) attack~\cite{baruch2019little}, and the Dissensus attack~\cite{he2022byzantine}, a new form of attack specifically designed for DFL systems.

\subsubsection{Comparison DFL Methods}
We evaluate the effectiveness of our proposed \alg by comparing it with the following eight methods. 
Note that FedAvg~\cite{McMahan17}, Krum~\cite{blanchard2017machine}, Trimmed Mean (Trim-mean)~\cite{Yin18}, Median~\cite{Yin18}, and FLTrust~\cite{cao2020fltrust} were originally designed for server-assisted FL, which are adapted to the DFL setting.

\myparatight{FedAvg~\cite{McMahan17}} 
In the FedAvg method, a client collects local models from its neighbor clients, then takes the weighted average of all collected models.

\myparatight{Krum~\cite{blanchard2017machine}} 
When client $i \in \mathcal{V}$ gets $\left| \mathcal{N}_i \right|$ local models from neighbors, it chooses the model closest in Euclidean distance to its $| \mathcal{N}_i |- \lceil c_i| \mathcal{N}_i | \rceil -2$ nearest models. Here, $\mathcal{N}_i$ is client $i$'s neighbor set, with $c_i$ and $\lceil c_i| \mathcal{N}_i | \rceil$ denoting the proportion and count of malicious neighbors, respectively.

\myparatight{Trimmed Mean (Trim-mean)~\cite{Yin18}} 
Once client $i \in \mathcal{V}$ receives $\left|\mathcal{N}_i\right|$ local models from its neighbors, it first removes the largest and smallest $\lceil c_i| \mathcal{N}_i | \rceil$ elements for each dimension, then computes the average of the rest.

\myparatight{Median~\cite{Yin18}} 
In the Median rule, each client $i$ computes the coordinate-wise median of all collected $\left|\mathcal{N}_i\right|$ local models.

\myparatight{FLTrust~\cite{cao2020fltrust}}In FLTrust, client \(i\) calculates the cosine similarity between its local model \(\bm{w}_i^{t+\frac{1}{2}}\) and a neighbor's model \(\bm{w}_j^{t+\frac{1}{2}}\) upon receipt. If this similarity is positive, \(\bm{w}_j^{t+\frac{1}{2}}\) is normalized to \(\tilde{\bm{w}}_j^{t+\frac{1}{2}}\) with the same magnitude as \(\bm{w}_i^{t+\frac{1}{2}}\), followed by client \(i\) averaging all normalized models received from neighbors.

\myparatight{UBAR~\cite{guo2021byzantine}}The UBAR aggregation rule employs a two-stage process to filter out any potentially malicious information. Specifically, during training round $t$, client $i$ first identifies a subset of neighbors $\mathcal{N}_i^s$, which consists of those with the smallest sum of distance to $\bm{w}_i^{t+\frac{1}{2}}$, where $\mathcal{N}_i^s \subseteq \mathcal{N}_i$, $| \mathcal{N}_i^s   | = | \mathcal{N}_i | - \lceil c_i| \mathcal{N}_i | \rceil$.
In the second stage, client $i$ narrows down the subset even further by selecting a new subset $\mathcal{N}_i^r$ from $\mathcal{N}_i^s$, which only includes neighbors whose loss values are smaller than its own loss. Finally, client $i$ averages the local models from $\mathcal{N}_i^r$.

\myparatight{LEARN~\cite{el2021collaborative}}In LEARN, clients exchange both local model updates and local models with their neighboring clients, and utilize Trim-mean aggregation rule to combine the received local model updates and local models.
Specifically, during training round $t$, client $i$ aggregates local model updates from its neighboring clients for $\lceil \log_2 t \rceil$ times, then exchanges local models with its neighbors once.

\myparatight{Self-Centered Clipping (SCCLIP)~\cite{he2022byzantine}} 
In the SCCLIP aggregation rule, each client clips all received local models from its neighbor clients based on its own local model.

\subsubsection{Evaluation Metrics}
For the synthetic dataset, we employ maximum mean squared error (Max.MSE) as the evaluation criterion, as we train a linear regression model on this synthetic dataset. For the four real-world datasets, we use maximum testing error rate (Max.TER) and maximum attack success rate (Max.ASR) as the evaluation metrics, as these datasets are used for training classification models. For all three evaluation metrics, smaller values indicate stronger defense capabilities.

\myparatight{Maximum mean squared error (Max.MSE)} 
In the linear regression model, we first calculate the mean squared error (MSE) for each benign client's final local model. The MSE is computed as \(\text{MSE} = \frac{1}{n_{\text{test}}} \sum_{i=1}^{n_{\text{test}}}(y_i - \hat{y}_i)^2\), where \(y_i\) is the actual value, \(\hat{y}_i\) denotes the predicted value, and \(n_{\text{test}}\) is the number of testing examples. We then assess a DFL method's robustness on the synthetic dataset by selecting the maximum MSE among all benign clients.

\myparatight{Maximum testing error rate (Max.TER)~\cite{guo2021byzantine}}
Following~\cite{guo2021byzantine}, we compute the testing error rate of the final local model on each benign client, and use the maximum testing error rate among all benign clients to measure the robustness of a DFL method.

\myparatight{Maximum attack success rate (Max.ASR)}
We compute the attack success rate of the final local model on each benign client, and report the maximum attack success rate among all benign clients.
The attack success rate is the fraction of targeted testing examples classified as the attacker-chosen targeted label.

\subsubsection{Non-IID Setting} 
Training data in FL are typically Non-IID (not independently and identically distributed) across clients.
In our paper, we consider the IID setting for synthetic dataset, and Non-IID setting for four real-world datasets.
We use the way in~\cite{fang2020local} to simulate the Non-IID setting for MNIST, Fashion-MNIST datasets.
In this approach, for a dataset containing $z$ classes, clients are first divided into $z$ random groups. A training example labeled $h$ is allocated to clients in group $h$ with a specific probability $p$, and to those in different groups with a probability of $\frac{1-p}{z-1}$. Within the same group, the training examples are evenly distributed among the clients. 
An increase in $p$ results in a greater level of Non-IID. In our experiment, we set $p=0.8$, indicating a substantial imbalance in the distribution of labels among clients. For example, 80\% of the training data for a client is concentrated in a single class.
In Section~\ref{sec:discussion_limitation}, we explore a more extreme Non-IID scenario where each client's training data is limited to just a few classes (e.g., three).
The HAR dataset's training data are inherently heterogeneous, eliminating the need for Non-IID simulation. Similarly, the CelebA dataset, processed as per~\cite{caldas2018leaf}, already exhibits Non-IID characteristics, so additional Non-IID simulation is unnecessary.

\subsubsection{Parameter Setting} 
We assume that there are a total of 20 clients for synthetic, MNIST, Fashion-MNIST, and CelebA datasets. 
Note that each smartphone user can be seen as a client in the HAR dataset, thus there are 30 clients in total for that dataset.
By default, we assume that 20\% of clients are malicious.
In our experiments, we train a linear regression model on the synthetic dataset.
Note that the population risk of the linear regression model satisfies Assumption~\ref{assumption_1} and Assumption~\ref{assumption_2}.
We use a convolutional neural network (CNN) to train the MNIST, Fashion-MNIST, and CelebA datasets. The architecture of CNN is shown in Table~\ref{cnn_arch} in Appendix. 
For the HAR dataset, we train a logistic regression classifier.
We train 300, 2,000, 2,000, 1,000 and 1,500 rounds for synthetic, MNIST, Fashion-MNIST, HAR and CelebA datasets, respectively.
The learning rates are respectively set to $6\times10^{-4}$, $3\times10^{-4}$, $6\times10^{-3}$, $3\times10^{-3}$ and $5\times10^{-5}$ for five datasets.
For all datasets, we set $\alpha =0.5$, $\gamma=0.3, \kappa=1, \lambda(t) = \frac{t}{T}$.

By default, we assume that all clients use the same initial local model, parameters $\alpha$, $\gamma, \kappa, \lambda(t)$, and aggregation rule $\text{AGG}$. We will also explore the settings where clients have different initial local models, $\alpha$, and $\text{AGG}$.
Aligned with existing work~\cite{pasquini2022privacy}, we consider regular graph as the default communication graph, where each node has an equal number of neighboring nodes.
We use regular-($n$, $v$) to denote a regular graph with $n$ nodes, where each node is connected to $v$ neighbors.
By default, we use a regular-(20, 10) graph for the synthetic, MNIST, Fashion-MNIST, and CelebA datasets. 
The HAR dataset inherently consists of 30 clients, so we consider a regular-(30, 15) graph structure for HAR.
Fig.~\ref{fig_no_label_figure_regular_n20_10deg} and Fig.~\ref{fig_no_label_figure_regular_n30_15deg} in Appendix show the topologies of regular-(20, 10) and regular-(30, 15) graphs, respectively.
Note that in Fig.~\ref{fig_diff_graph}, each 
node represents a client.
The nodes highlighted in red indicate malicious clients, while the nodes highlighted in blue represent benign clients.
By default, we assume that the communication graph $\mathcal{G}$ is static, i.e., the edges between clients will not change over time.  
We will also explore the time-varying communication graph setting, where each client has certain possibility of not sharing information with its neighboring clients in each round.
We perform experiments on four NVIDIA Tesla V100 GPUs, repeating each experiment 10 times and averaging the results.
Default results are reported for the MNIST dataset using a regular-(20, 10) graph, with 4 out of 20 clients being malicious.

\begin{table}[htbp]
	\centering
\scriptsize
\setlength{\doublerulesep}{3\arrayrulewidth}
     \caption{Results of different DFL methods. The results of Backdoor are presented as ``Max.TER / Max.ASR''.}
    \label{tab_all_datasets_app}%
     \subfloat[MNIST dataset.]
    {
    	\begin{tabular}{|c|c|c|c|c|c|c|c|c|c|}
        \hline
        Method & No & LF & Feature & Gauss  & Krum & Trim & Backdoor & Adapt \\
        \hline
        \hline
        FedAvg & 0.10 & 0.10 & 0.90 & 0.90 & 0.91 & 0.91 & 0.90 / 1.00 & 0.90 \\
        Krum & 0.10 & 0.12 & 0.90 & 0.10 & 0.10 & 0.15 & 0.15 / 0.01 & 0.14  \\
        Trim-mean & 0.11 & 0.12 & 0.49 & 0.11 & 0.82 & 0.81 & 0.83 / 0.72 & 0.87  \\
        Median & 0.14 & 0.14 & 0.45 & 0.15 & 0.52 & 0.63 & 0.66 / 0.01 & 0.66 \\
        FLTrust & 0.10 & 0.11 & 0.90 & 0.13 & 0.10 & 0.88 & 0.10 / 0.73 & 0.10 \\
        UBAR & 0.14 & 0.14 & 0.90 & 0.14 & 0.14 & 0.14 & 0.15 / 0.01 & 0.14 \\
        LEARN & 0.10 & 0.10 & 0.30 & 0.12 & 0.18 & 0.57 & 0.12 / 0.03 & 0.44 \\
        SCCLIP & 0.10 & 0.10 & 0.10 & 0.11 & 0.91 & 0.91 & 0.10 / 0.01 & 0.91 \\
        \rowcolor{greyL}
        \alg & 0.10 & 0.10 & 0.11 & 0.10 & 0.10 & 0.11 & 0.11 / 0.01 & 0.11  \\
        \hline
    \end{tabular}%
     \label{tab_all_datasets_app_mnist}%
    }
     \quad
    \subfloat[Fashion-MNIST dataset.]
    {
    	\begin{tabular}{|c|c|c|c|c|c|c|c|c|c|}
        \hline
        Method & No & LF & Feature & Gauss & Krum & Trim & Backdoor & Adapt \\
        \hline
        \hline
        FedAvg & 0.16 & 0.21 & 0.90 & 0.90 & 0.90 & 0.90 & 0.90 / 1.00 & 0.90 \\
        Krum & 0.26 & 0.27 & 0.90 & 0.40 & 0.37 & 0.48 & 0.28  / 0.03 & 0.42  \\
        Trim-mean & 0.25 & 0.27 & 0.90 & 0.27 & 0.77 & 0.87 & 0.90 / 1.00 & 0.76 \\
        Median & 0.26 & 0.26 & 0.90 & 0.28 & 0.54 & 0.74 & 0.90 / 1.00 & 0.69   \\
        FLTrust & 0.19 & 0.20 & 0.19 & 0.90 & 0.25 & 0.90 & 0.19 / 0.99 & 0.90  \\
        UBAR & 0.21 & 0.23 & 0.90 & 0.21 & 0.22 & 0.22 & 0.24 / 0.03 & 0.23  \\
        LEARN & 0.23 & 0.26 & 0.47 & 0.23 & 0.34 & 0.37 & 0.23 / 0.90 & 0.51  \\
        SCCLIP & 0.20 & 0.25 & 0.90 & 0.33 & 0.89 & 0.89 & 0.90 / 1.00 & 0.52 \\
        \rowcolor{greyL}
        \alg & 0.16 & 0.17 & 0.17 & 0.16 & 0.17 & 0.17 & 0.17 / 0.02 & 0.17  \\
        \hline
    \end{tabular}%
    \label{tab_all_datasets_app_Fashion}%
    }
      \quad
		\subfloat[HAR dataset.]
	{
    	\begin{tabular}{|c|c|c|c|c|c|c|c|c|c|}
	\hline
	Method & No & LF & Feature & Gauss & Krum & Trim & Backdoor & Adapt  \\
	\hline
	\hline
        FedAvg & 0.04 & 0.04 & 0.45 & 0.98 & 0.32 & 0.38 & 0.82 / 1.00 & 0.99 \\
        Krum & 0.10 & 0.10 & 0.10 & 0.10 & 0.10 & 0.10 & 0.10 / 0.01 & 0.10 \\
        Trim-mean & 0.05 & 0.06 & 0.06 & 0.06 & 0.09 & 0.17 & 0.08 / 0.01 & 0.09 \\
        Median & 0.06 & 0.06 & 0.08 & 0.06 & 0.07 & 0.17 & 0.07 / 0.01 & 0.08 \\
        FLTrust & 0.04 & 0.04 & 0.08 & 0.07 & 0.04 & 0.31 & 0.04 / 0.47 & 0.05\\
        UBAR & 0.06 & 0.06 & 0.06 & 0.06 & 0.06 & 0.12 & 0.08 / 0.03 & 0.06  \\
        LEARN & 0.04 & 0.04 & 0.04 & 0.04 & 0.06 & 0.13 & 0.05 / 0.04 & 0.06\\
        SCCLIP & 0.05 & 0.05 & 0.13 & 0.07 & 0.27 & 0.32 & 0.06 / 0.02 & 0.12 \\
	\rowcolor{greyL}
	\alg & 0.04 & 0.05 & 0.04 & 0.05 & 0.04 & 0.05 & 0.04 / 0.01 & 0.05 \\
	\hline
	\end{tabular}%
      \label{tab_all_datasets_app_HAR}%
	}
     \quad
	\subfloat[CelebA dataset.]
	{
    	\begin{tabular}{|c|c|c|c|c|c|c|c|c|c|}
		\hline
		Method & No & LF & Feature & Gauss  & Krum & Trim & Backdoor & Adapt\\
		\hline
		\hline
		FedAvg & 0.10 & 0.16 & 0.48 & 0.48 & 0.53 & 0.53 & 0.48 / 0.01 & 0.48 \\
            Krum & 0.18 & 0.26 & 0.48 & 0.30 & 0.31 & 0.18 & 0.20 / 0.26 & 0.18\\
            Trim-mean & 0.12 & 0.24 & 0.35 & 0.15 & 0.15 & 0.26 & 0.22 / 0.09 & 0.19 \\
            Median & 0.13 & 0.17 & 0.31 & 0.15 & 0.15 & 0.26 & 0.21 / 0.15 & 0.19 \\
            FLTrust & 0.10 & 0.14 & 0.10 & 0.11 & 0.11 & 0.53 & 0.12 / 0.06 & 0.10 \\
            UBAR & 0.12 & 0.13 & 0.48 & 0.13 & 0.12 & 0.14 & 0.14 / 0.13 & 0.13 \\
            LEARN & 0.31 & 0.41 & 0.37 & 0.35 & 0.51 & 0.53 & 0.32 / 0.09 & 0.53  \\
            SCCLIP & 0.10 & 0.17 & 0.14 & 0.12 & 0.43 & 0.53 & 0.48 / 0.01 & 0.53 \\
		\rowcolor{greyL}
		\alg & 0.10 & 0.12 & 0.11 & 0.11 & 0.11 & 0.11 & 0.12 / 0.02 & 0.13  \\
		\hline
	\end{tabular}%
       \label{tab_all_datasets_app_CelebA}%
	}
	\vspace{-.15in}
\end{table}%

\begin{figure*}[!t]
	\centering
	\includegraphics[scale = 0.4]{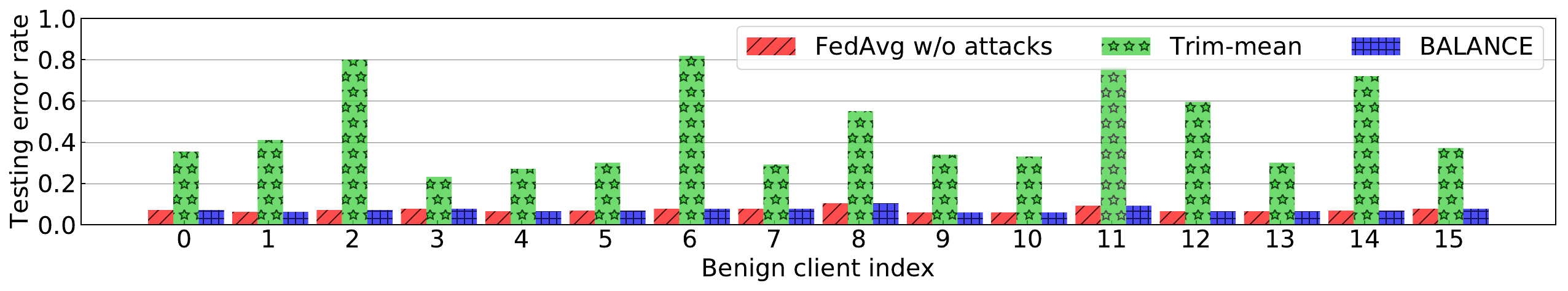}
	\caption{Testing error of each benign client of FedAvg without any attacks, Trim-mean aggregation rule and our proposed method under Trim attack.}
	\label{each_client_error_rate_trim}
		\vspace{-.1in}
\end{figure*}

\begin{figure}[!t]
	\centering
\subfloat[Communication cost.]{\includegraphics[width=0.24 \textwidth]{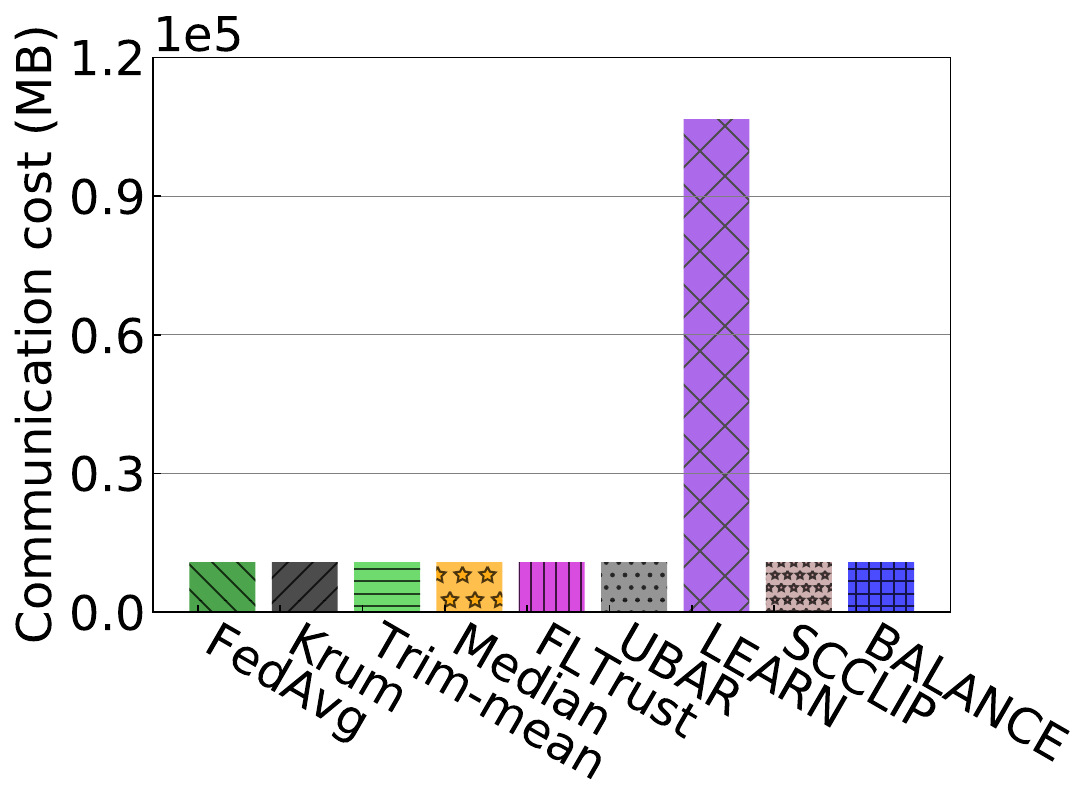}}
	\subfloat[Computation cost.]{\includegraphics[width=0.24 \textwidth]{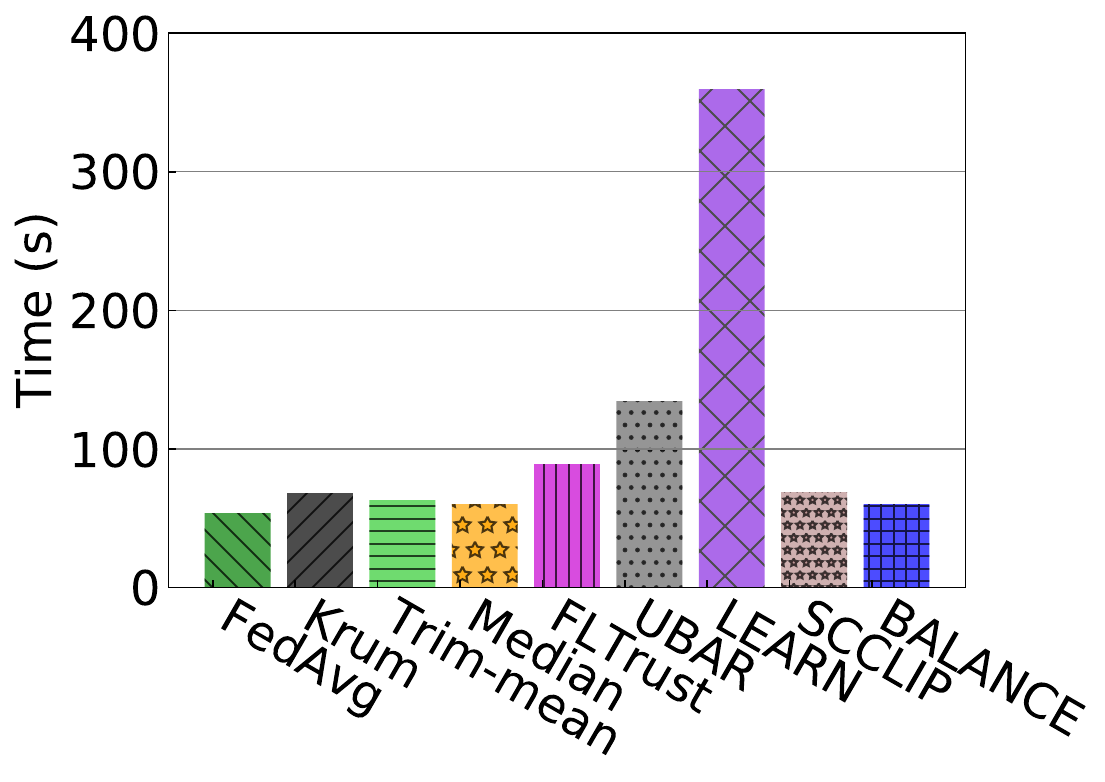}} 
 \caption{Communication and computation costs of different methods.
}
 \label{cost_fig}
\vspace{-.1in}
\end{figure}

\begin{figure}[!t]
	\centering
	\includegraphics[width=0.48\textwidth]{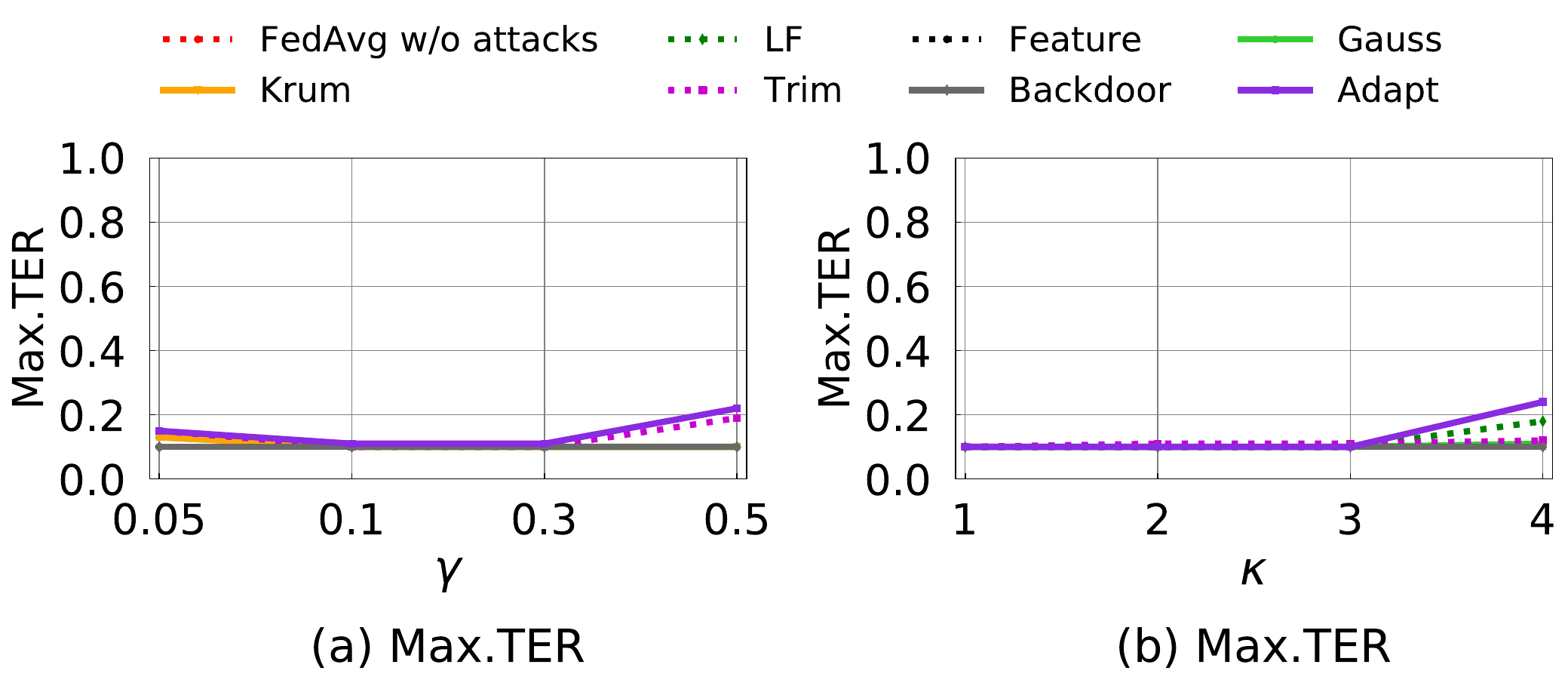}
	\caption{Impact of $\gamma$ and $\kappa$.}
	\label{diff_gamma}
		\vspace{-.1in}
\end{figure}

\subsection{Experimental Results}

\begin{figure*}[!t]
	\centering
	\includegraphics[scale = 0.46]{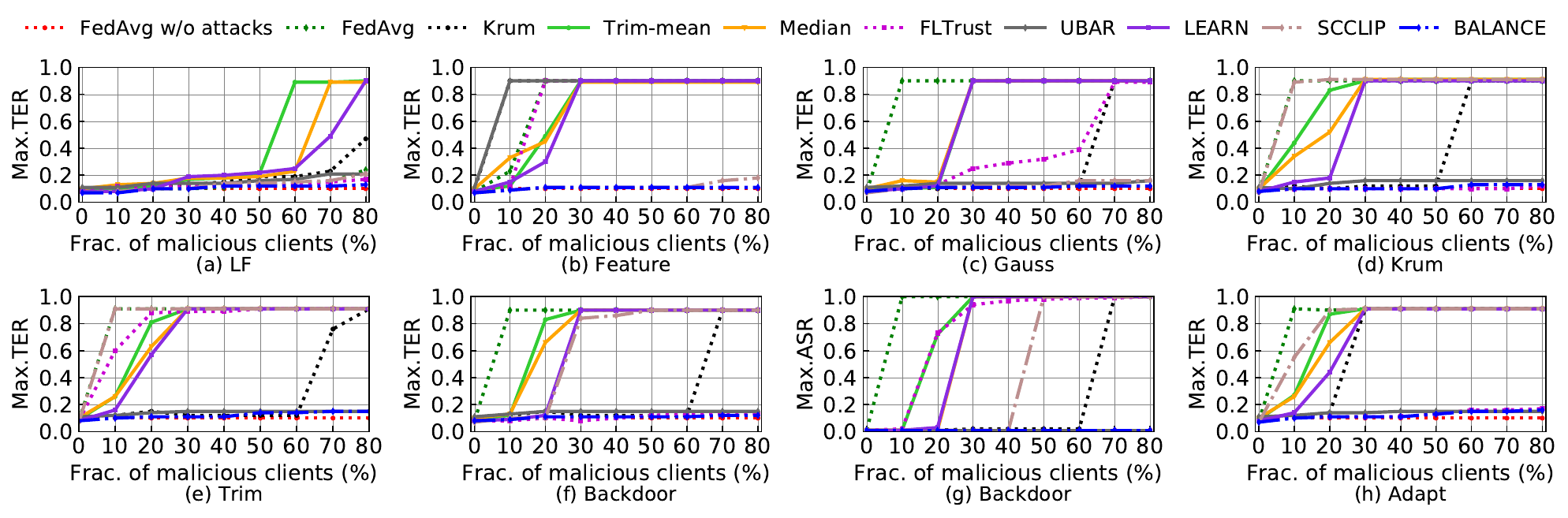}
	\caption{Impact of fraction of malicious clients. 
}
	\label{fig_attack_size_max_error}
		\vspace{-.1in}
\end{figure*}

\begin{figure*}[!t]
	\centering
	\includegraphics[scale = 0.46]{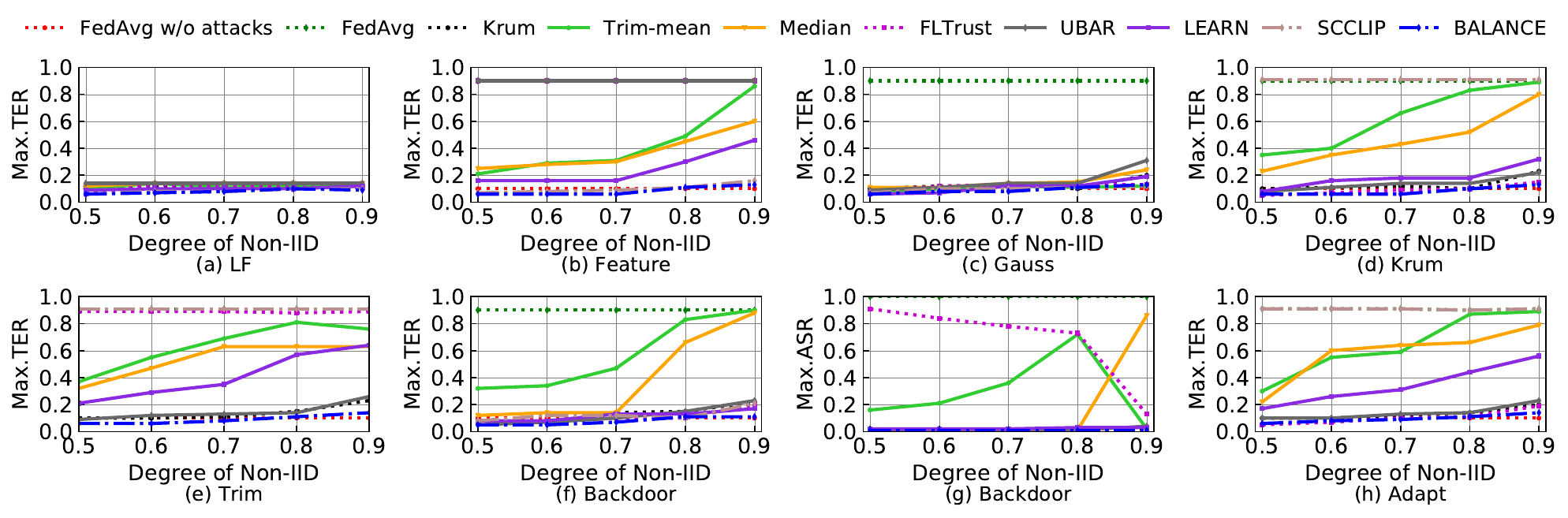}
	\caption{Impact of degree of Non-IID. 
}
	\label{non_iid_max_error}
		\vspace{-.1in}
\end{figure*}

\myparatight{Our proposed \alg is effective}
We first demonstrate the effectiveness of our proposed \alg on synthetic dataset, where  the population risk is both $\mu$-strongly and $L$-smooth, i.e., satisfying Assumption~\ref{assumption_1} and Assumption~\ref{assumption_2}.
Table~\ref{tab_all_datasets_app_Synthetic} in Appendix shows the results of different methods under different attacks on synthetic dataset. 
Each row corresponds to a different DFL method.
``No'' means all clients are benign.
We exclude Backdoor attacks for the synthetic dataset as there are no specific Backdoor attacks for regression models.
From Table~\ref{tab_all_datasets_app_Synthetic}, we observe that our proposed method outperforms baselines in both non-adversarial and adversarial scenarios, the Max.MSEs of our method are comparable to those of FedAvg without attacks.

Next, we show the performance of our method on four real-world datasets, where the trained models are highly non-convex, results are shown in Table~\ref{tab_all_datasets_app}.
The results for the Backdoor attack are given as``Max.TER / Max.ASR''.
We note that DFL method achieves comparable performance to its server-assisted counterpart.
For instance, when FedAvg aggregation rule is used and all clients are benign, the test error of the final global model is 0.09 in server-assisted FL.
We also remark that in DFL, clients could not obtain accurate final models when they independently train their models without exchanging information with other clients.
On MNIST dataset, the Max.TER is 0.29 when clients solely train models locally.

First, we observe that when there is no attack, i.e., all clients are benign, our proposed \alg achieves similar Max.TER as that of FedAvg under no attack. This means that our method achieves the ``competitive learning performance'' goal mentioned in Section~\ref{problem_statement}. For instance, on the CelebA dataset, both our proposed method and FedAvg under no attack exhibit a Max.TER of 0.10. However, the Max.TER
 of LEARN is 0.31, see Table~\ref{tab_all_datasets_app_CelebA}.
Next, we find that our proposed \alg is resilient to different types of poisoning attacks, including data poisoning and model poisoning attacks, and performs better than existing methods. 
For instance, on the MNIST dataset, Trim-mean's Max.TER increases from 0.11 to 0.81 under the Trim attack. In contrast, our method maintains a small  corresponding Max.TER of 0.11.
We observe similar results on the other three real-world datasets, indicating that our \alg achieves the ``Byzantine robustness'' goal.
We remark that \alg either matches or outperforms all existing methods known to converge to (global) optimal points. This shows \alg does not get stuck at non- or local-optimal points.
We also note that the Adapt attack demonstrates the most effective attack performance when targeting our proposed method, whereas it may perform worse when attacking other methods. The reason is that the Adapt attack is specifically designed for our method.

Fig.~\ref{each_client_error_rate_trim} shows the testing error rate of each benign client, when clients utilize FedAvg without any attacks, and when they use Trim-mean aggregation rule and our proposed \alg under Trim attack.
We observe that under Trim attack, the testing error rate of each benign client's final learned model escalates substantially when the clients adopt the Trim-mean aggregation rule to merge the local models from neighboring clients. However, our proposed method guarantees that each benign client will obtain a final model that is almost as accurate as FedAvg without any attacks.

Fig.~\ref{cost_fig} shows the communication and computation costs of various methods when we train the FL system for 2,000 rounds on the MNIST dataset, where regular-(20, 10) communication graph is used. More specifically, for a given FL method, the communication cost refers to the size of data (local model or local model update) that each client sends to its neighboring clients over 2,000 rounds, while the computation cost indicates the time that each client requires to aggregate the received local models (updates) over 2,000 rounds.
As seen in Fig.~\ref{cost_fig}, our \alg demonstrates both communication and computation efficiency. 
Conversely, other methods lead to high communication and computation costs. For example, in each training round of the LEARN method, each client must first exchange local model updates with its neighboring clients for $\lceil \log_2 t \rceil$ rounds, and then exchanges local model once. 
This information exchange process incurs significant communication and computation costs.

From Table~\ref{tab_all_datasets_app} and Table~\ref{tab_all_datasets_app_Synthetic}, we also observe that the application of server-assisted FL methods to DFL results in suboptimal performance. Specifically, FLTrust is particularly prone to Trim attack on the MNIST dataset. This vulnerability arises from FLTrust's underlying assumption that the server's root dataset mirrors the distribution of the clients' overall training data, an assumption that often does not hold in practical FL scenarios. Moreover, the training data of clients in DFL are highly heterogeneous. As a result, when a client employs FLTrust for aggregation, it tends to incorrectly classify many benign neighboring clients as malicious.

Table~\ref{tab_avg_error} in Appendix presents consensus errors~\cite{he2022byzantine,lian2017can,kong2021consensus} for various methods. Consensus error measures the average squared difference between each benign client's final model and all benign clients' average model. 
The details of this metric is shown in Section~\ref{sec:more_metrics} in Appendix. 
Our method shows low values in this  metric. We also remark that consensus error alone cannot determine whether the final learned model is accurate or not. 
A small consensus error may also imply that all benign clients have reached a poor consensus, meaning that all benign clients have learned similar but inaccurate models.
Thus we omit the consensus error metric in subsequent experiments.
Note that we also do not consider the average testing error rate (Avg.TER) of benign clients, since low Avg.TER can sometimes mask high errors in individual clients.

\myparatight{Impact of $\gamma$ and $\kappa$}
Fig.~\ref{diff_gamma} shows the results of our proposed \alg under various poisoning attacks with different values of $\gamma$ and $\kappa$.
We observe that the Max.TER of \alg is large when $\gamma$ and $\kappa$ are too large.
The reason is that for our method, a client would falsely reject local models shared by benign neighboring clients when $\kappa$ is too large, as a large $\kappa$ leads to a rapid decrease in $\gamma \exp(-\kappa \cdot  \lambda(t))$. The local models from malicious neighboring clients may get accepted if $\gamma$ is too large.

\begin{figure*}[!t]
	\centering
	\includegraphics[scale = 0.46]{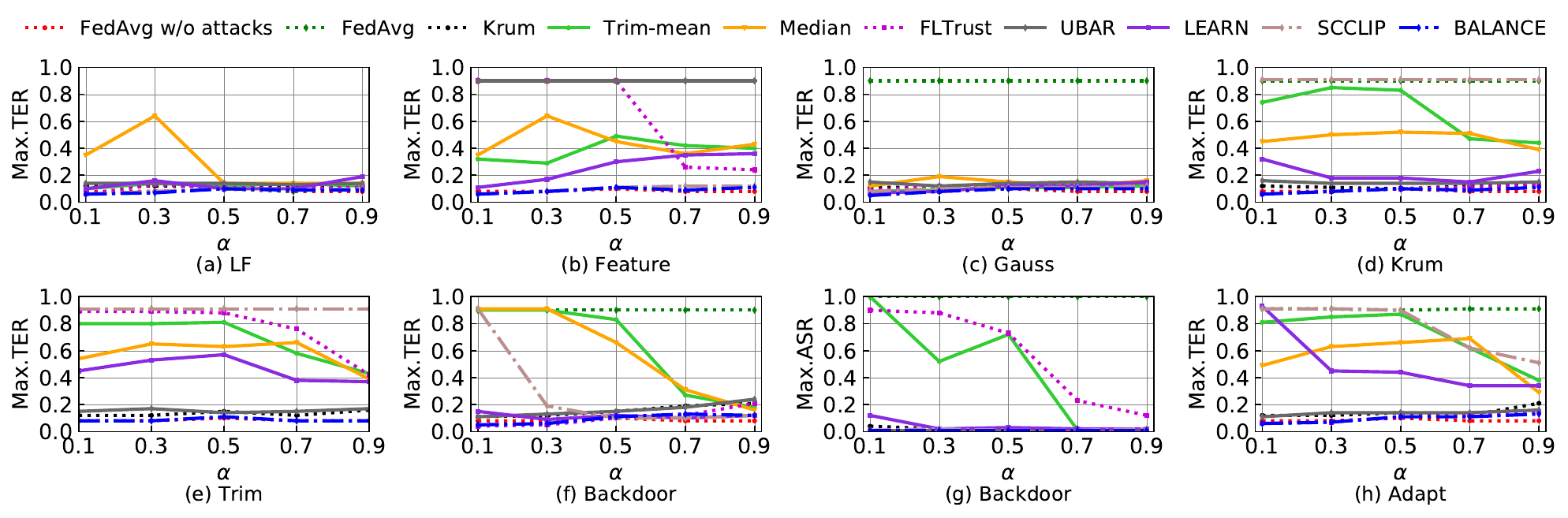}
	\caption{Impact of $\alpha$. 
}
	\label{alpha_max_error}
		\vspace{-.2in}
\end{figure*}

\myparatight{Impact of fraction of malicious clients}
Fig.~\ref{fig_attack_size_max_error} shows the results of different methods under different attacks on MNIST dataset and regular-(20, 10) communication graph, when the fraction of malicious clients varies from 0\% to 80\%, and the total number clients is set to 20.
We observe that our proposed DFL approach is the only method that can withstand 50\% of malicious clients, while existing Byzantine-robust methods lead to significant Max.TERs even when only a small proportion of clients are malicious. For example, UBAR aggregation rule is susceptible to poisoning attacks when only 10\% of clients are malicious, as seen in the case of the Feature attack strategy, where the maximum testing error rate rises to 0.90. Furthermore, our proposed approach can withstand even the most powerful Adapt attack when 80\% of clients are malicious.

\myparatight{Impact of degree of Non-IID} Fig.~\ref{non_iid_max_error} displays the results of different methods under poisoning attacks with varying degrees of Non-IID.
We observe that our proposed method outperforms existing DFL methods with all considered Non-IID scenarios. 
For example, when the degree of Non-IID is relatively low, such as 0.5, the Trim attack on the Median aggregation rule leads to a Max.TER of 0.32. However, for our proposed method, the Max.TERs of all benign clients are not significantly high, even with highly heterogeneous training data across all clients.

\myparatight{Impact of $\alpha$}
According to Eq.~(\ref{local_model_agg}), in DFL, each client utilizes a parameter called $\alpha$ to balance the combination of its local model and those of its neighboring clients. A higher value of $\alpha$ indicates greater trust in the client's own local model, while a smaller value of $\alpha$ results in more weight placed on the aggregated neighboring models.
Note that by default, we assume that all clients use the same $\alpha$.
Fig.~\ref{alpha_max_error} shows the results of different methods under various poisoning attacks, when we vary the value of $\alpha$, where other parameters are set to their default settings.
We can observe that in a particular DFL method, when the poisoning attack is weak, clients can achieve greater model accuracy by setting a smaller value of the trade-off parameter $\alpha$. This means that clients can benefit more by collaborating with others and giving more weight to the models received from their neighboring clients. For example, for our proposed \alg under Gauss attack, the Max.TERs are 0.05 and 0.10 when $\alpha$ is set to 0.1 and 0.9, respectively.
However, when the attack is strong, a smaller value of $\alpha$ may result in a larger Max.TER. This is because benign clients may receive malicious models from their neighboring clients, and if the DFL method is not robust, giving more weight to the neighboring models through a smaller $\alpha$ may lead to a larger Max.TER. For example, under Trim attack, when $\alpha=0.1$, the Max.TER of the Trim-mean method is 0.80, while when $\alpha$ is set to 0.9, the Max.TER is 0.43. 
Note that in the extreme case where $\alpha=1$, which means each client trains its own local model independently without sharing information with others. In our experiment, we find that the Max.TER is 0.29 when each client only uses its own local training data to train the model. 

The paper~\cite{pasquini2022privacy} employs a model aggregation technique where each client aggregates its model as $\bm{w}_i^{t+1} = \text{AGG} \{     \bm{w}_j^{t+\frac{1}{2}}, j \in \widehat{\mathcal{N}_i} \}$, where $\widehat{\mathcal{N}_i} = \mathcal{N}_i \cup \{ i\}$, $\mathcal{N}_i$ is the set of neighbors of client $i$ (not including client $i$ itself). Our proposed method is compared to existing DFL methods using this aggregation setting, and the results are shown in Table~\ref{tab_aggre_all} in Appendix. 
We observe that our proposed method is also robust against various poisoning attacks and outperforms baseline methods 
under this  setting.
We also observe that if clients aggregate their models using the setting suggested in~\cite{pasquini2022privacy}, the Max.TERs of existing defenses without attacks are very large.
For instance, the Max.TERs of Krum and UBAR are respectively 0.18 and 0.25 when all clients are benign, however, the corresponding Max.TERs are 0.10 and 0.14 (see the results in Table~\ref{tab_all_datasets_app_mnist}), respectively when clients perform aggregation based on Eq.~(\ref{local_model_agg}).

\myparatight{Clients use different $\alpha$ or different aggregation rules}
By default, in our experiments, clients use the same $\alpha$ and follow the same aggregation rule to combine their local models with their neighboring clients' models. 
In this section, we first investigate the scenario where different clients use different $\alpha$ values. 
In our experiments, each client randomly samples its $\alpha$ value from the interval [0, 1]. We consider two cases: Case I and Case II. In Case I, each client randomly samples its $\alpha$ value from the interval [0, 1] before the training process begins, and then $\alpha$ remains fixed for that client throughout the training process. 
In Case II, each client randomly samples its $\alpha$ value from the interval [0, 1] in each training round, i.e., $\alpha$ changes during training for each client.
Note that in Case I and Case II, all clients still use the same aggregation rule.
The results for Case I are shown in Table~\ref{tab_diff_alpha} in Appendix, the results for 
Case II are shown in Table~\ref{tab_diff_alpha_app} in Appendix.
Comparing Table~\ref{tab_all_datasets_app_mnist} and Table~\ref{tab_diff_alpha}, we observe that 
having different clients use different values of $\alpha$ cannot reduce the impact of poisoning attacks.
Moreover, when all clients are benign, the Max.TERs of DFL methods including our \alg 
 are large compared to the scenario when clients using the same $\alpha$.
Our proposed \alg achieves similar Max.TERs under different attacks compared to the FedAvg method without any attacks.
However, existing defenses are still vulnerable to poisoning attacks.

We then study the scenario where different clients use different aggregation rules to combine the received neighboring clients' models. We also investigate two cases, namely Case III and Case IV. In both cases, we randomly assign one existing Byzantine-robust aggregation rule from set $\mathcal{H}=$\{Krum, Trim-mean, Median, FLTrust, UBAR, SCCLIP\}
to each client. Note that the set $\mathcal{H}$ excludes FedAvg, LEARN, and our proposed method. This is because FedAvg is not robust; LEARN method requires exchanging both local model updates and local models between clients, while other methods only need to exchange local models; and our proposed method is already robust and does not require being used with other aggregation rules.
Specifically, in Case III, each client $i \in \mathcal{V}$ randomly selects one aggregation rule from the set $\mathcal{H}$ before the training process.
In Case IV, in each training round, client $i$ randomly selects one robust Byzantine-robust aggregation rule from the set $\mathcal{H}$.
The results are shown in Table~\ref{tab_diff_agg} in Appendix.
Compare Table~\ref{tab_all_datasets_app_mnist} and Table~\ref{tab_diff_agg}, we observe that when different aggregation rules are used by clients, they are more susceptible to poisoning attacks compared to the scenario where all clients use the same aggregation rule.

\begin{table}[t]
    \centering
    \scriptsize
    \setlength{\doublerulesep}{3\arrayrulewidth}
    \caption{Results of different DFL methods with time-varying communication graph.}
    \label{tab_diff_drop_prob}%
\begin{tabular}{|c|c|c|c|c|c|c|c|c|}
        \hline
        Method & No & LF & Feature & Gauss & Krum & Trim & Backdoor & Adapt \\
        \hline
        \hline
        FedAvg & 0.10 & 0.15 & 0.90 & 0.90 & 0.91 & 0.90 & 0.90 / 1.00 & 0.90 \\
        Krum & 0.13 & 0.15 & 0.90 & 0.90 & 0.15 & 0.14 & 0.90 / 1.00 & 0.90  \\
        Trim-mean & 0.27 & 0.27 & 0.90 & 0.90 & 0.91 & 0.91 & 0.90 / 1.00 & 0.90  \\
        Median & 0.24 & 0.28 & 0.90 & 0.90 & 0.91 & 0.91 & 0.90 / 1.00 & 0.90  \\
        FLTrust & 0.11 & 0.16 & 0.90 & 0.89 & 0.89 & 0.89 & 0.89 / 0.09 & 0.89  \\
        UBAR & 0.15 & 0.17 & 0.90 & 0.90 & 0.15 & 0.19 & 0.15 / 0.01 & 0.90 \\
        LEARN & 0.19 & 0.19 & 0.38 & 0.90 & 0.91 & 0.91 & 0.90 / 1.00 & 0.90  \\
        SCCLIP & 0.11 & 0.16 & 0.34 & 0.15 & 0.91 & 0.91 & 0.23 / 0.02 & 0.64   \\
        \rowcolor{greyL}
        \alg & 0.11 & 0.12 & 0.12 & 0.12 & 0.11 & 0.12 & 0.11 / 0.01 & 0.12 \\
        \hline
    \end{tabular}%
\end{table}%

\myparatight{Time-varying communication graph}
By default, we consider a static communication graph $\mathcal{G}$, i.e., once established, $\mathcal{G}$ is fixed. 
Here we consider a practical setting, where each client has a chance of disconnecting from the protocol, such as Internet issues.
When a client disconnects from the protocol during a specific round, it is unable to exchange information with its neighboring clients.
However, a disconnected client can continue training its model locally, and it may reconnect to the protocol in the subsequent round.
We consider the default parameter settings, where there are 4 out of 20 clients are malicious, MNIST dataset and regular-(20, 10) graph are used (client may disconnect from the protocol based on the regular-(20, 10) graph).
However, each client has 20\% possibility of disconnecting from the protocol.
The results are shown in Table~\ref{tab_diff_drop_prob}.
We observe that existing defenses are more vulnerable to poisoning attacks, while \alg could still defend against various attacks.

\myparatight{Impact of the total number of clients}
Fig.~\ref{number_of_clients_max_error} in Appendix shows the results of different methods on MNIST dataset when the total number of clients is changed, while keeping the fraction of malicious clients fixed at 20\%.
Note that we use regular-(10, 5), regular-(20, 10), regular-(30, 15), regular-(40, 20) and regular-(50, 25) communication graphs when we have 10, 20, 30, 40 and 50 clients in total, respectively.
The topologies of regular-(20, 10) and regular-(30, 15) graphs are shown in Figs.~\ref{fig_no_label_figure_regular_n20_10deg}-\ref{fig_no_label_figure_regular_n30_15deg}.
The topologies of regular-(10, 5), regular-(40, 20) and regular-(50, 25) graphs are illustrated in Figs.~\ref{fig_no_label_figure_regular_n10_5deg}-\ref{fig_no_label_figure_regular_n50_25deg}.
We observe from Fig.~\ref{number_of_clients_max_error} that our proposed DFL approach is resilient to poisoning attacks for all the considered total number of clients ranging from 10 to 50.

\myparatight{Different initial models or communication graphs}
We assume that all clients use the same initial local model by default. We also explore the setting where different clients use different initial local models, other parameters are set to their default settings.
The results are shown in Table~\ref{tab_diff_init} in Appendix.
We observe that our proposed method is robust against poisoning attacks and outperforms baselines, even when clients use different initial models.

We also investigate other types of communication graphs, including complete graph, Erdős–Rényi graph, small-world graph and ring graph, 
 the topologies of these four graphs are illustrated in Figs.~\ref{fig_no_label_figure_adjacency_full_20}-\ref{fig_ring_graph} in Appendix. 
Note that for a complete graph, each client is connected to the remaining clients; for a ring graph, clients form a ring, and each client only has two neighbors.
In all four communication graphs, there are 20 clients in total, and 4 clients are malicious.
We maintain default settings for other parameters. 
The results for four graphs are shown in Table~\ref{tab_diff_graph_app} in Appendix.
We observe that our \alg can also achieve Byzantine-robustness when other types of communication graphs are used.

\myparatight{Impact of fraction of edges between malicious and benign clients}
Malicious clients in a DFL system attempt to manipulate the system by sharing harmful information, such as carefully crafted local models, with their neighboring clients. 
The attack performance is generally influenced by the \textit{Fraction of Edges between Malicious and Benign clients (FEMB)}.
In our paper, FEMB is calculated as the ratio of the number of edges between malicious and benign clients to the total number of edges in the communication graph.
In our experiments, we generate three random graphs to study the impact of FEMB on the attack performance.
For all three random graphs, there are 20 nodes in total, each node represents one client, and 4 out of 20 clients are malicious.
The FEMBs for the three graphs are 0.16, 0.22, and 0.32, respectively.
The topologies of three graphs are shown in Figs.~\ref{fig_FEMB_016}-\ref{fig_FEMB_032}.
All other parameters are kept at their default values.
The experimental results are shown in Table~\ref{tab_diff_FEMB_results_app} in Appendix.
We observe that the attack performance generally increases with an increase in FEMB.


\section{Discussion and Limitations} 
\label{sec:discussion_limitation}

\begin{table}[t]
	\centering
	\scriptsize
	\setlength{\doublerulesep}{3\arrayrulewidth}
	\caption{Results of different DFL methods, where each client only has three classes of training data.}
	\label{tab_noniid_three}
	\begin{tabular}{|c|c|c|c|c|c|c|c|c|}
		\hline
		Method & No & LF & Feature & Gauss & Krum & Trim & Backdoor & Adapt \\
		\hline
		\hline
		FedAvg & 0.10 & 0.18 & 0.90 & 0.91 & 0.91 & 0.90 & 0.90 / 1.00 & 0.91 \\
		Krum & 0.46 & 0.49 & 0.91 & 0.49 & 0.57 & 0.52 & 0.47 / 0.16 & 0.46\\
		Trim-mean & 0.17 & 0.62 & 0.27 & 0.24 & 0.58 & 0.83 & 0.79 / 0.23 & 0.81 \\
		Median & 0.56 & 0.89 & 0.56 & 0.88 & 0.64 & 0.70 & 0.56 / 0.04 & 0.58 \\
		FLTrust & 0.10 & 0.18 & 0.91 & 0.22 & 0.18 & 0.90 & 0.20 / 0.27 & 0.14 \\
		UBAR & 0.42 & 0.42 & 0.91 & 0.52 & 0.52 & 0.57 & 0.62 / 0.32 & 0.54  \\
		LEARN & 0.10 & 0.15 & 0.10 & 0.15 & 0.20 & 0.61 & 0.16 / 0.09 & 0.58  \\
		SCCLIP & 0.10 & 0.18 & 0.23 & 0.21 & 0.90 & 0.91 & 0.33 / 0.08 & 0.89 \\
		\rowcolor{greyL}
		\alg & 0.10 & 0.14 & 0.13 & 0.14 & 0.14 & 0.14 & 0.11 / 0.02 & 0.14 \\
		\hline
	\end{tabular}%
    \vspace{-.10in}
\end{table}%

\myparatight{More extreme Non-IID distribution}
In our default Non-IID setting, a client's primary training data come from just one class, with few examples from other classes. This section explores a more extreme Non-IID scenario as described in~\cite{McMahan17}. Here, the distribution of training data among clients is based solely on labels, and each client having training data from only three classes. For instance, Client 1 possesses training data only for labels 0-2, while Client 2 exclusively holds data for labels 3-5.

The results of various DFL methods under various attacks with this more extreme Non-IID scenario are shown in Table~\ref{tab_noniid_three}. When comparing Table~\ref{tab_all_datasets_app_mnist} with Table~\ref{tab_noniid_three}, it becomes apparent that in this Non-IID context, the Max.TER of current DFL methods is significantly high, even in non-adversarial setting. For instance, with the UBAR aggregation rule and all clients being benign, the Max.TER reaches 0.42. In adversarial setting, our suggested \alg continues to effectively counter all the poisoning attacks considered, achieving the largest Max.TER of just 0.14. Nonetheless, existing DFL approaches show increased susceptibility to poisoning attacks. To illustrate, the Max.TER of Median is 0.89 under LF attack.

\begin{table}[t]
	\centering
	\scriptsize
	\caption{Results of different DFL methods under LIE and Dissensus attacks.}
	\label{tab_lie_discent_attack}%
	\begin{tabular}{|c|c|c|c|}
		\hline
		Method & No & LIE & Dissensus \\
		\hline
		\hline
		FedAvg & 0.10 & 0.13 & 0.91 \\
		Krum & 0.10 & 0.10 & 0.15  \\
		Trim-mean & 0.11 & 0.16 & 0.86\\
		Median & 0.14 & 0.18 & 0.87 \\
		FLTrust & 0.10 & 0.10 & 0.12 \\
		UBAR & 0.14 & 0.14 & 0.23  \\
		LEARN & 0.10 & 0.10 & 0.69\\
		SCCLIP & 0.10 & 0.10 & 0.91  \\
		\rowcolor{greyL}
		\alg    & 0.10 & 0.10 & 0.10 \\
		\hline
	\end{tabular}%
    \vspace{-.10in}
\end{table}%

\myparatight{More adaptive and decentralized attacks}In our prior experiments, we show that our proposed \alg is robust against data poisoning attacks (such as LF and Feature attacks) as well as sophisticated adaptive attacks. For these data poisoning attacks, attackers corrupt the local training data on malicious clients. While the models trained on this poisoned data might appear benign, our experiments reveal that such attacks are ineffectual against our \alg due to their limited impact. In the case of the adaptive attacks we considered, the attacker, aware of our aggregation rule, introduces subtle, strategic perturbation to the benign local models in order to circumvent our defenses. Nonetheless, our experimental results indicate that as attackers adapt their strategies to evade our defense mechanism, the efficacy of their attacks diminishes.

In this section, we first explore a different type of adaptive attack known as the ``a little is enough'' (LIE) attack~\cite{baruch2019little}. The LIE attack represents a general attack model where the attacker is not required to be aware of the aggregation rule employed by other clients. In executing a LIE attack, the attacker calculates the variance among benign models, and then introduces minimal changes to these models, intended to circumvent the aggregation rule. Our current experiments demonstrate that extending server-based attacks to the DFL setting can be effective against prevailing DFL methods. Additionally, we examine the Dissensus attack~\cite{he2022byzantine}, a newly devised attack model specifically tailored for DFL systems. In the Dissensus attack, the attacker designs malicious local models with the intention of disrupting consensus among benign clients.

Table~\ref{tab_lie_discent_attack} presents the results of various DFL methods under LIE and Dissensus attacks. It is evident from the results that the LIE attack does not significantly compromise existing DFL methods. This could be attributed to the fact that the LIE attack is a more generic attack model and is not specifically tailored for a fully decentralized environment. In contrast, the Dissensus attack shows a considerable ability to disrupt DFL methods, particularly in the case of Trim-mean and Median. For instance, under the Dissensus attack, the Max.TER for the Median method escalates to 0.87.

\begin{table}[t]
    \centering
    \scriptsize
      \setlength{\doublerulesep}{3\arrayrulewidth}
     \caption{Results of different variants of BALANCE.}
    \label{tab_variant}
\begin{tabular}{|c|c|c|c|c|c|c|c|c|}
        \hline
        Variant & No & LF & Feature & Gauss & Krum & Trim & Backdoor & Adapt \\
        \hline
        \hline
        Variant I   & 0.10 & 0.11 & 0.12 & 0.11 & 0.11 & 0.11 &0.11 / 0.01 &0.14 \\
        Variant II   & 0.10 &0.10  &  0.13& 0.16 &  0.11& 0.18  & 0.11 / 0.01 & 0.16  \\
        \rowcolor{greyL}
        \alg & 0.10 & 0.10 & 0.11 & 0.10 & 0.10 & 0.11 & 0.11 / 0.01 & 0.11 \\
        \hline
    \end{tabular}%
	\vspace{-.1in}
\end{table}%

\myparatight{Different variants of \algns}
In this part, we consider two variants of our proposed \algns.

\begin{list}{\labelitemi}{\leftmargin=1em \itemindent=-0.5em \itemsep=.2em}
	\item \textbf{Variant I:} In this variant, client \( i \) accepts \( \bm{w}_j^{t+\frac{1}{2}} \) when the condition \( \| \bm{w}_i^{t+\frac{1}{2}} - \bm{w}_j^{t+\frac{1}{2}} \| \le \gamma \| \bm{w}_i^{t+\frac{1}{2}} \| \) holds true.

	\smallskip
	\item \textbf{Variant II:} Client $i$ calculates \( q_j = \frac{\| \bm{w}_i^{t+\frac{1}{2}} - \bm{w}_j^{t+\frac{1}{2}} \|}{\| \bm{w}_i^{t+\frac{1}{2}} \|} \) for each neighbor \( j \) in its neighbor set \( \mathcal{N}_i \). After computing the median of these values from the \( |\mathcal{N}_i| \) neighbors, denoted as \( q_\text{med} \) where \( q_\text{med}=\text{med}\{q_1,...,q_{|\mathcal{N}_i|} \} \), client \( i \) will accept \( \bm{w}_j^{t+\frac{1}{2}} \) if \( q_j \le \min\{q_\text{med}, \gamma\} \). The \( \gamma \) is set to safeguard against the possibility of most neighbors being malicious. The key idea of Variant II is that client $i$ will accept client $j$'s local model if it is close to its own, based on a comparison with the median of deviations from all neighbors.
\end{list}

Table~\ref{tab_variant} compares our \alg with two variants. Variant II underperforms due to its tendency to incorrectly reject many benign local models. In contrast, Variant I shows performance comparable to our \algns. 
Note that the distance between \( \bm{w}_i^{t+\frac{1}{2}} \) and \( \bm{w}_j^{t+\frac{1}{2}} \) will become smaller, as the DFL system approaches convergence, leading Variant I to inadvertently accept some malicious models in later training stages. Nevertheless, attacking the DFL system becomes challenging as the model nears convergence. 
Additional experiments support this claim, showing that the Max.TER of Median under Trim attack is 0.55 when attacks occur only in the first half of training rounds (1-1,000 rounds out of 2,000 total), while it is 0.38 when attacks only occur in the second half. When attacks happen in all rounds, the Max.TER is 0.63.
In our \algns, training the model for a sufficient number of rounds (a large $T$) can significantly reduce the value of \(\gamma \cdot \exp(-\kappa \cdot  \lambda(t))\). By selecting a smaller \(\kappa\), we can decelerate the decline of the exponential function.


\section{Conclusion and Future Work} 
\label{sec:conclusion}

In this work, we proposed a novel method called \alg to defend against poisoning attacks in DFL.
In our proposed method, each client uses its local model as a reference point to check whether the received neighboring client's local model is malicious or benign.
We established the convergence performance of our method under poisoning attacks in both strongly convex and non-convex settings, and the convergence rate of our \alg matches those of the state-of-the-art counterparts in Byzantine-free settings.
Extensive experiments across various settings demonstrated the efficacy of our proposed method.
Our future work includes designing an optimized strategy to dynamically select aggregation rules and parameter $\alpha$ for different clients to enhance the robustness of existing DFL methods.

\begin{acks}
We thank the anonymous reviewers for their comments. 
This work was supported by NSF grants CAREER CNS-2110259, CNS-2112471, CNS-2312138, SaTC-2350075, No. 2131859, 2125977, 2112562, 1937786, 1937787, and ARO grant No. W911NF2110182.
\end{acks}

\balance
\bibliographystyle{ACM-Reference-Format}
\bibliography{refs}


\appendix



\section{Proof of Theorem~\ref{Theorem1}} \label{sec:appendix_1}

We denote by $\bm{g}(\bm{w}_{i}^t)$ the stochastic gradient that client $i$ computed based on the model $\bm{w}_{i}^t$, then we have that:
\begin{align}
\label{client_i_sgd}
\bm{w}_{i}^{t  +\frac{1}{2}} = \bm{w}_{i}^{t } - \eta\bm{g}(\bm{w}_{i}^{t}),
\end{align}
where $\eta>0$ is the learning rate.

In the following proof, we ignore the superscript $t$ in $ \mathcal{S}_i^{t}$ for simplicity.
Thus we have that:
\begin{align}
\label{client_i_agg_minus_w_i_t}
& \bm{w}_{i}^{t+1} - \bm{w}_{i}^{t} \nonumber \\
&\stackrel{(a)} =\alpha \bm{w}_{i}^{t+\frac{1}{2}} + (1-\alpha)  \text{AGG} \{     \bm{w}_j^{t+\frac{1}{2}}, j \in \mathcal{N}_i \} - \bm{w}_{i}^{t}  \nonumber \\
&\stackrel{(b)} = \alpha \bm{w}_{i}^{t+\frac{1}{2}}  +  (1-\alpha)  \frac{1}{| \mathcal{S}_i|} \sum\limits_{j \in \mathcal{S}_i}  (   \bm{w}_j^{t+\frac{1}{2}} -  \bm{w}_{i}^{t+\frac{1}{2}} +  \bm{w}_{i}^{t+\frac{1}{2}} ) - \bm{w}_{i}^{t} \nonumber \\
&= [\alpha + (1-\alpha)] \bm{w}_{i}^{t+\frac{1}{2}} + \frac{1-\alpha}{| \mathcal{S}_i|} \sum\limits_{j \in \mathcal{S}_i}  (   \bm{w}_j^{t+\frac{1}{2}} -  \bm{w}_{i}^{t+\frac{1}{2}} )  - \bm{w}_{i}^{t}   \nonumber \\
&= \bm{w}_{i}^{t+\frac{1}{2}} + \frac{1-\alpha}{| \mathcal{S}_i|} \sum\limits_{j \in \mathcal{S}_i}  (   \bm{w}_j^{t+\frac{1}{2}} -  \bm{w}_{i}^{t+\frac{1}{2}} )  - \bm{w}_{i}^{t}   \nonumber \\
&\stackrel{(c)} =  [\bm{w}_{i}^t -  \eta \bm{g}(\bm{w}_{i}^{t}) - \bm{w}_{i}^{t}] + \frac{1-\alpha}{| \mathcal{S}_i|} \sum\limits_{j \in \mathcal{S}_i}  (   \bm{w}_j^{t+\frac{1}{2}} -  \bm{w}_{i}^{t+\frac{1}{2}} )  \nonumber \\
&= - \eta \bm{g}(\bm{w}_{i}^{t}) + \frac{1-\alpha}{| \mathcal{S}_i|} \sum\limits_{j \in \mathcal{S}_i}  (   \bm{w}_j^{t+\frac{1}{2}} -  \bm{w}_{i}^{t+\frac{1}{2}} ),
\end{align}
where $(a)$ is because of Eq.~(\ref{local_model_agg}); $(b)$ is due to Eq.~(\ref{our_agg_combine}); $(c)$ is due to Eq.~(\ref{client_i_sgd}).

According to Assumption~\ref{assumption_2}, when the loss function is $L$-smooth, one has that:
\begin{align}
\label{smooth_xy}
F(\bm{w}_{i}^{t+1}) \le F(\bm{w}_{i}^t) + \left\langle  \nabla F(\bm{w}_{i}^t), \bm{w}_{i}^{t+1} - \bm{w}_{i}^t \right\rangle  \nonumber \\
+ \frac{L}{2} \left\| \bm{w}_{i}^{t+1} - \bm{w}_{i}^t  \right\|^2.
\end{align}

Combining Eq.~(\ref{client_i_agg_minus_w_i_t}) and Eq.~(\ref{smooth_xy}), we obtain:
\begin{align}
\label{smooth_xy_first_equ}
& F(\bm{w}_{i}^{t+1})  \nonumber \\
& \le  F(\bm{w}_{i}^t) + \langle  \nabla F(\bm{w}_{i}^t), - \eta \bm{g}(\bm{w}_{i}^{t}) + \frac{1-\alpha}{| \mathcal{S}_i|} \sum\limits_{j \in \mathcal{S}_i}  (   \bm{w}_j^{t+\frac{1}{2}} -  \bm{w}_{i}^{t+\frac{1}{2}} ) \rangle \nonumber \\
& \quad + \frac{L}{2} \| - \eta \bm{g}(\bm{w}_{i}^{t}) + \frac{1-\alpha}{| \mathcal{S}_i|} \sum\limits_{j \in \mathcal{S}_i}  (   \bm{w}_j^{t+\frac{1}{2}} -  \bm{w}_{i}^{t+\frac{1}{2}} )  \|^2 \nonumber \\ 
& \stackrel{(a)} \le  F(\bm{w}_{i}^t) - \eta \langle \nabla F(\bm{w}_{i}^t), \bm{g}(\bm{w}_{i}^{t}) \rangle   \nonumber \\
& \quad + \langle  \nabla F(\bm{w}_{i}^t),  \frac{1-\alpha}{| \mathcal{S}_i|} \sum\limits_{j \in \mathcal{S}_i}  (   \bm{w}_j^{t+\frac{1}{2}} -  \bm{w}_{i}^{t+\frac{1}{2}} )  \rangle   +  L \eta^2  \| \bm{g}(\bm{w}_{i}^{t})  \|^2 \nonumber \\
& \quad + L \| \frac{1-\alpha}{| \mathcal{S}_i|} \sum\limits_{j \in \mathcal{S}_i}  (   \bm{w}_j^{t+\frac{1}{2}} -  \bm{w}_{i}^{t+\frac{1}{2}} )   \|^2,
\end{align}
where $(a)$ is because of Lemma~\ref{Lem_Sum_Norm} in Appendix~\ref{sec:appendix_3}.

Taking expectation on both sides of Eq.~(\ref{smooth_xy_first_equ}), one obtains:
\begin{align}
& \EX [ F(\bm{w}_{i}^{t+1})] \nonumber \\
& \stackrel{(a)} \le \EX [  F(\bm{w}_{i}^{t}) - \eta  \| \nabla F(\bm{w}_{i}^t)  \|^2 \nonumber \\
& \quad + (1-\alpha)\langle \nabla F(\bm{w}_{i}^t),  \frac{1}{| \mathcal{S}_i|} \sum\limits_{j \in \mathcal{S}_i}  (   \bm{w}_j^{t+\frac{1}{2}} -  \bm{w}_{i}^{t+\frac{1}{2}} )  \rangle    
\nonumber \\
&  \quad +  L \eta^2  \| \bm{g}(\bm{w}_{i}^{t}) - \nabla F(\bm{w}_i^t) + \nabla F(\bm{w}_i^t)\|^2 \nonumber\\
&  \quad + L (1-\alpha)^2 \| \frac{1}{| \mathcal{S}_i|} \sum\limits_{j \in \mathcal{S}_i}  (   \bm{w}_j^{t+\frac{1}{2}} -  \bm{w}_{i}^{t+\frac{1}{2}} )   \|^2 ]  \nonumber \\
& \overset{(b)}{\leq}  \EX [  F(\bm{w}_{i}^{t}) - \eta  \| \nabla F(\bm{w}_{i}^t)  \|^2 + 2 L \eta^2 \| \nabla F(\bm{w}_{i}^t)  \|^2   \nonumber \\
& \quad + (1-\alpha)\langle \nabla F(\bm{w}_{i}^t),  \frac{1}{| \mathcal{S}_i|} \sum\limits_{j \in \mathcal{S}_i}  (   \bm{w}_j^{t+\frac{1}{2}} -  \bm{w}_{i}^{t+\frac{1}{2}} )  \rangle    
\nonumber \\
& \quad +  2 L \eta^2  \delta ^2 + L (1-\alpha)^2 \| \frac{1}{| \mathcal{S}_i|} \sum\limits_{j \in \mathcal{S}_i}  (   \bm{w}_j^{t+\frac{1}{2}} -  \bm{w}_{i}^{t+\frac{1}{2}} )   \|^2 ] \nonumber \\
& = \EX [   F(\bm{w}_{i}^{t}) -(\eta- 2L \eta^2)\| \nabla F(\bm{w}_{i}^t)  \|^2  +2 L \eta^2 \delta ^2  \nonumber \\
&\quad + (1-\alpha)\langle \nabla F(\bm{w}_{i}^t),  \frac{1}{| \mathcal{S}_i|} \sum\limits_{j \in \mathcal{S}_i}  (   \bm{w}_j^{t+\frac{1}{2}} -  \bm{w}_{i}^{t+\frac{1}{2}} )  \rangle  \nonumber \\
& \quad  +  L (1-\alpha)^2 \| \frac{1}{| \mathcal{S}_i|} \sum\limits_{j \in \mathcal{S}_i}  (   \bm{w}_j^{t+\frac{1}{2}} -  \bm{w}_{i}^{t+\frac{1}{2}} )   \|^2 ] \nonumber \\
&  \stackrel{(c)} \le \EX [  F(\bm{w}_{i}^t)  -  \eta(1 - 2L \eta)  \| \nabla F(\bm{w}_{i}^t) \|^2 + 2 L \eta_t^2 \delta^2 \nonumber \\
& \quad +  (1-\alpha) \|  \nabla F(\bm{w}_{i}^t) \| \cdot \| \frac{1}{\left| \mathcal{S}_i \right|} \sum\limits_{j \in \mathcal{S}_i} (\bm{w}_j^{t+\frac{1}{2}} - \bm{w}_{i}^{t+\frac{1}{2}} )\| \nonumber \\
& \quad + L (1-\alpha)^2\|  \frac{1}{| \mathcal{S}_i |} \sum\limits_{j \in \mathcal{S}_i} (\bm{w}_j^{t+\frac{1}{2}} - \bm{w}_{i}^{t+\frac{1}{2}} ) \|^2  ], 
\end{align}
where $(a)$ is because of Assumption~\ref{assumption_3} that $\EX [\bm{g}(\bm{w}_{i}^{t}) ] = \nabla F(\bm{w}_{i}^t)$; $(b)$ is due to Lemma~\ref{Lem_Sum_Norm} in Appendix~\ref{sec:appendix_3}, and Assumption~\ref{assumption_3} that $\EX [ \|  \bm{g} (\bm{w}_{i}^t) -  \nabla F(\bm{w}_{i}^t) \|]^2 \le  \delta ^2 $; $(c)$ is because of Cauchy-Schwarz inequality.

Next, we bound the term $\| \frac{1}{| \mathcal{S}_i |}  \sum\limits_{j \in \mathcal{S}_i} (\bm{w}_j^{t+\frac{1}{2}} - \bm{w}_{i}^{t+\frac{1}{2}} )  \|$:
\begin{align}
&\| \frac{1}{| \mathcal{S}_i |}  \sum\limits_{j \in \mathcal{S}_i} (\bm{w}_j^{t+\frac{1}{2}} - \bm{w}_{i}^{t+\frac{1}{2}} ) \|
=  \|  \frac{1}{| \mathcal{S}_i |}  \sum\limits_{j \in \mathcal{S}_i} (\bm{w}_j^{t+\frac{1}{2}} - \bm{w}_{i}^{t+\frac{1}{2}} )  \| \nonumber \\
& \le  \frac{1}{| \mathcal{S}_i |}  \sum\limits_{j \in \mathcal{S}_i} \| \bm{w}_j^{t+\frac{1}{2}} - \bm{w}_{i}^{t+\frac{1}{2}}    \| 
 \stackrel{(a)} \le \frac{1}{| \mathcal{S}_i |}  \sum\limits_{j \in \mathcal{S}_i} \gamma  \|\bm{w}_{i}^{t+\frac{1}{2}}  \| 
 \stackrel{(b)} \le \gamma  \psi, \nonumber
\end{align}
where $(a)$ is due to Eq.~(\ref{our_agg}) and \(\gamma \cdot \exp(-\kappa \cdot \lambda(t)) \le \gamma\); $(b)$ is because of Assumption~\ref{assumption_4} that $\| \bm{w}_{i}^t \|  \le \psi$.

Due to Assumption~\ref{assumption_4}, then we have:
\begin{align}
\|  \nabla F(\bm{w}_{i}^t) \| \cdot \| \frac{1}{\left| \mathcal{S}_i \right|} \sum\limits_{j \in \mathcal{S}_i} (\bm{w}_j^{t+\frac{1}{2}} - \bm{w}_{i}^{t+\frac{1}{2}} ) \| \le \gamma \psi \rho .
\end{align}

If $\gamma$ satisfies $\gamma \le \frac{\rho }{L\psi(1-\alpha)}$, then we have $ L \gamma ^2 \psi^2 (1-\alpha)^2 \le \gamma \rho\psi (1-\alpha)$, and we further have:
\begin{align}
\EX [ F(\bm{w}_{i}^{t+1})   ] 
& \le  \EX [  F(\bm{w}_{i}^t)  -  \eta (1 - 2 L \eta)  \| \nabla F(\bm{w}_{i}^t) \|^2  \nonumber \\
& \quad + 2 L \eta^2 \delta^2 + 2\gamma \rho\psi  (1-\alpha) ].
\end{align}

If the learning rate $\eta$ satisfies $\eta \le \frac{1}{4L}$, one has that:
\begin{align}
\label{non_convec_part1}
\EX [ F(\bm{w}_{i}^{t+1})   ] 
& \le  \EX [  F(\bm{w}_{i}^t)  - \frac{\eta_t}{2}  \| \nabla F(\bm{w}_{i}^t) \|^2 + 2 L \eta^2 \delta^2  \nonumber \\
& \quad + 2\gamma \rho\psi  (1-\alpha)].
\end{align}

By Assumption~\ref{assumption_1}, since $F(\cdot)$ is $\mu$-strongly convex, then one has the following Polyak-Łojasiewicz (PL) inequality~\cite{polyak1963gradient}
$
\|  \nabla F(\bm{w}_{i}^t)   \|^2 \ge 2\mu (  F(\bm{w}_{i}^t)  - F(\bm{w}^*) ).
$
Therefore, we further have that:
\begin{align}
\EX [ F(\bm{w}_{i}^{t+1})   ] 
& \le  \EX [  F(\bm{w}_{i}^t)  - \mu \eta  (  F(\bm{w}_{i}^t)  - F(\bm{w}^* ))  \nonumber \\
& \quad + 2 L \eta^2 \delta^2 +2\gamma \rho\psi  (1-\alpha)  ].
\end{align}

Subtracting $F(\bm{w}^*)$ on both sides, we get:

\begin{align}
\EX [ F(\bm{w}_{i}^{t+1})  - F(\bm{w}^*) ] 
& \le  (1-\mu \eta) \EX [   F(\bm{w}_{i}^t)  - F(\bm{w}^*) ]   \nonumber \\
& \quad + 2 L \eta^2 \delta^2 +2\gamma \rho\psi  (1-\alpha) ].
\end{align}

By choosing $\eta \le \frac{1}{\mu}$, then we can guarantee that $1-\mu \eta\ge 0$. Telescoping over $t=0, 1,...,T-1$, one has that:
\begin{align}
&  \EX [ F(\bm{w}_{i}^{t+1})  -  F(\bm{w}^*) ]   \nonumber \\
&\le  (1-\mu \eta )^T  [   F(\bm{w}_{i}^0)  -  F(\bm{w}^*) ] + \sum\limits_{t=0}^{T-1} (1-\mu \eta)^T  2 L \eta^2 \delta^2  \nonumber \\
& \quad+ \sum\limits_{t=0}^{T-1} (1-\mu \eta)^T 2\gamma \rho\psi  (1-\alpha) \nonumber \\
&= (1-\mu \eta)^T  [   F(\bm{w}_{i}^0)  - F(\bm{w}^*) ] + \frac{ 2 L \eta \delta^2}{\mu}
 + \frac{2 \gamma   \rho \psi (1-\alpha)}{\mu \eta}. \nonumber
\end{align}

\section{Proof of Theorem~\ref{Theorem2}} \label{sec:appendix_2}

According to Eq.~(\ref{non_convec_part1}), we have the following:
\begin{align}
\EX [ F(\bm{w}_{i}^{t+1})   ] 
& \le  \EX [  F(\bm{w}_{i}^t)  - \frac{\eta_t}{2}  \| \nabla F(\bm{w}_{i}^t) \|^2 + 2 L \eta^2 \delta^2  \nonumber \\
& \quad + 2\gamma \rho\psi  (1-\alpha)].
\end{align}

Rearranging the term, one obtains that:
\begin{align}
\frac{\eta_t}{2} \EX [  \| \nabla F(\bm{w}_{i}^t) \|^2 ] 
& \le  \EX [ F(\bm{w}_{i}^{t}) -   F(\bm{w}_{i}^{t+1}) ] + 2 L \eta^2 \delta^2  \nonumber \\
& \quad + 2\gamma \rho\psi  (1-\alpha) ].
\end{align}

By telescoping over $t=0, 1,...,T-1$ and taking an average over $T$, then we get:
\begin{align}
\frac{1}{T}  \sum\limits_{t=0}^{T-1}   \EX [   \| \nabla F(\bm{w}_{i}^t) \|^2 ] 
& \le \frac{2  [   F(\bm{w}_{i}^0)  -  F(\bm{w}^*) ]}{\eta T}   + 4 L \eta \delta^2  \nonumber \\
&\quad + \frac{4 \gamma   \rho\psi (1-\alpha)}{\eta}. 
\end{align}

\section{Useful Technical Lemma} \label{sec:appendix_3}

\begin{lem}
\label{Lem_Sum_Norm}
Given a set of vectors $\bm{x}_0, \bm{x}_1, \ldots, \bm{x}_{n-1}$ with $\bm{x}_i \in \mathbb{R}^d$ for all $i \in \{0,1,\cdots,n-1\}$, we have the following:
\begin{align*}
   \| \sum_{i = 0}^{n-1} \bm{x}_i \|^2 \leq n \sum_{i = 0}^{n-1} \| \bm{x}_i \|^2.
\end{align*}
\end{lem}

\begin{table}[t]
	\caption{The CNN architecture.}
	\centering
	\setlength{\doublerulesep}{3\arrayrulewidth}
	\vspace{1mm}
	\footnotesize 
	\begin{tabular}{|c|c|} \hline 
		{Layer} & {Size} \\ \hline \hline
		{Input} & { $28\times28\times1$}\\ \hline
		{Convolution + ReLU} & { $3\times3\times30$}\\ \hline
		{Max Pooling} & { $2\times2$}\\ \hline
		{Convolution + ReLU} & { $3\times3\times50$}\\ \hline
		{Max Pooling} & { $2\times2$}\\ \hline
		{Fully Connected + ReLU} & {100}\\ \hline
		{Softmax} & {10}\\ \hline
	\end{tabular}
	\label{cnn_arch}
\end{table}

\begin{figure*}[!t]
	\centering
	\captionsetup[subfloat]{textfont=scriptsize}
	\subfloat[Regular-(20, 10) graph.]{\includegraphics[width=0.152 \textwidth]{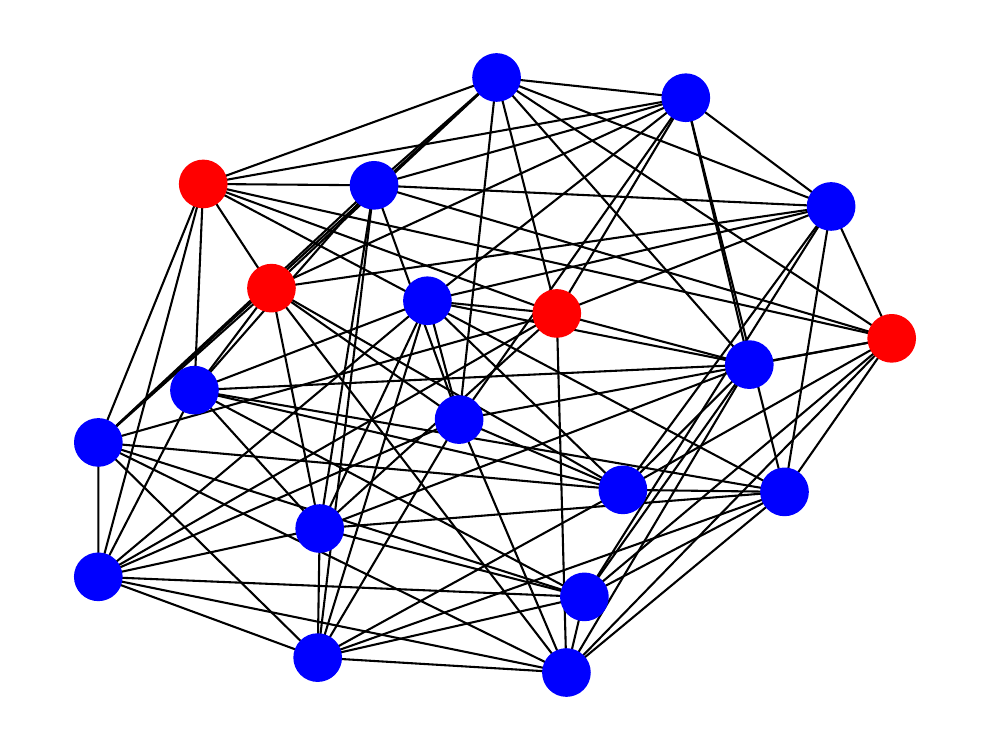}\label{fig_no_label_figure_regular_n20_10deg}}
	\subfloat[Regular-(30, 15) graph.]{\includegraphics[width=0.152 \textwidth]{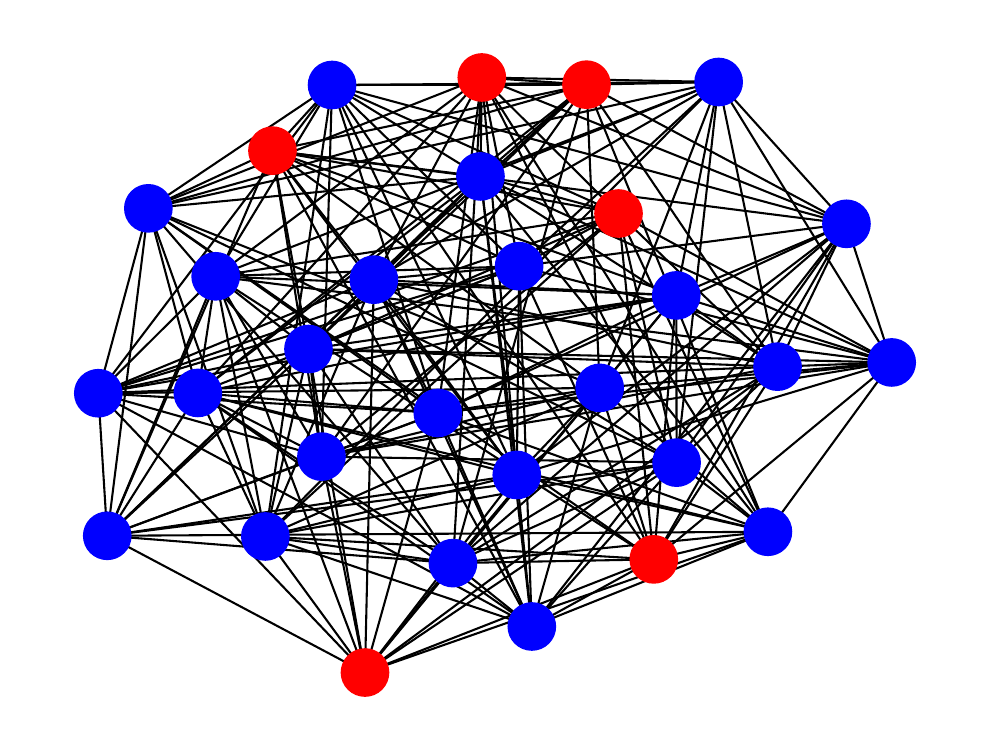}\label{fig_no_label_figure_regular_n30_15deg}} 
	\subfloat[Regular-(10, 5) graph.]{\includegraphics[width=0.152 \textwidth]{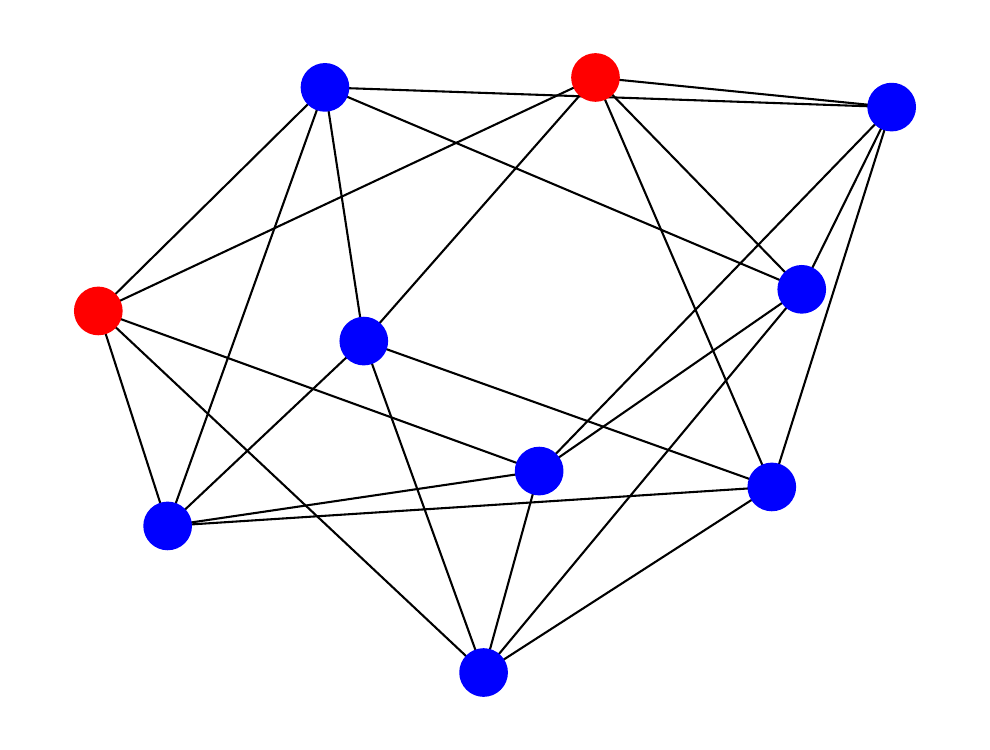}\label{fig_no_label_figure_regular_n10_5deg}}
	\subfloat[Regular-(40, 20) graph.]{\includegraphics[width=0.152 \textwidth]{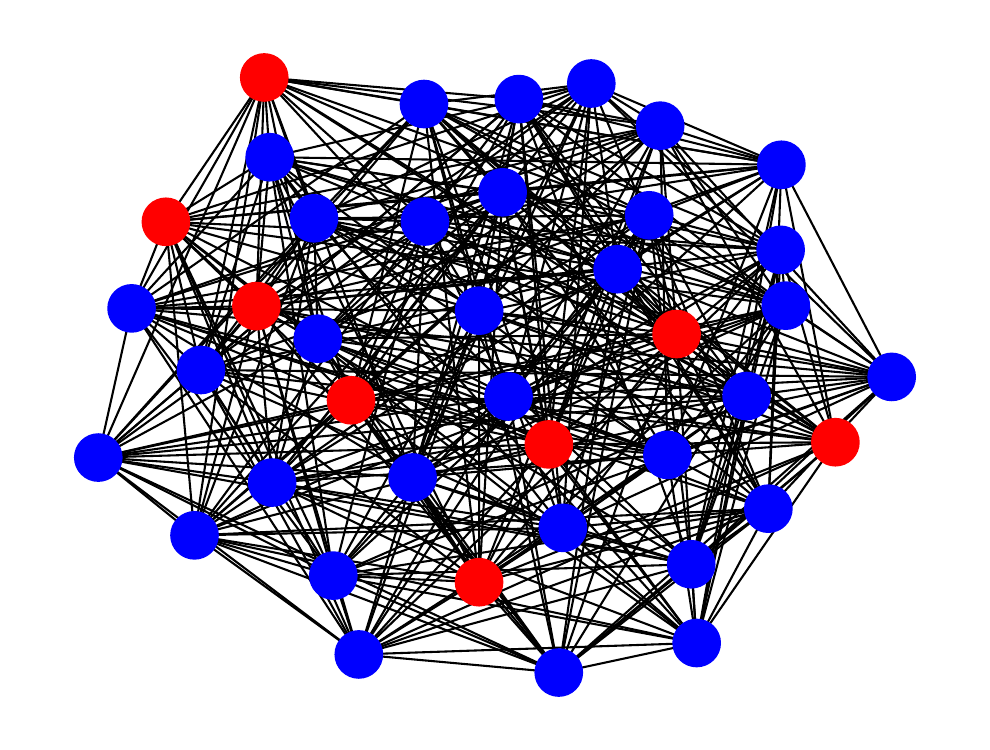}\label{fig_no_label_figure_regular_n40_20deg}} 
	\subfloat[Regular-(50, 25) graph.]{\includegraphics[width=0.152 \textwidth]{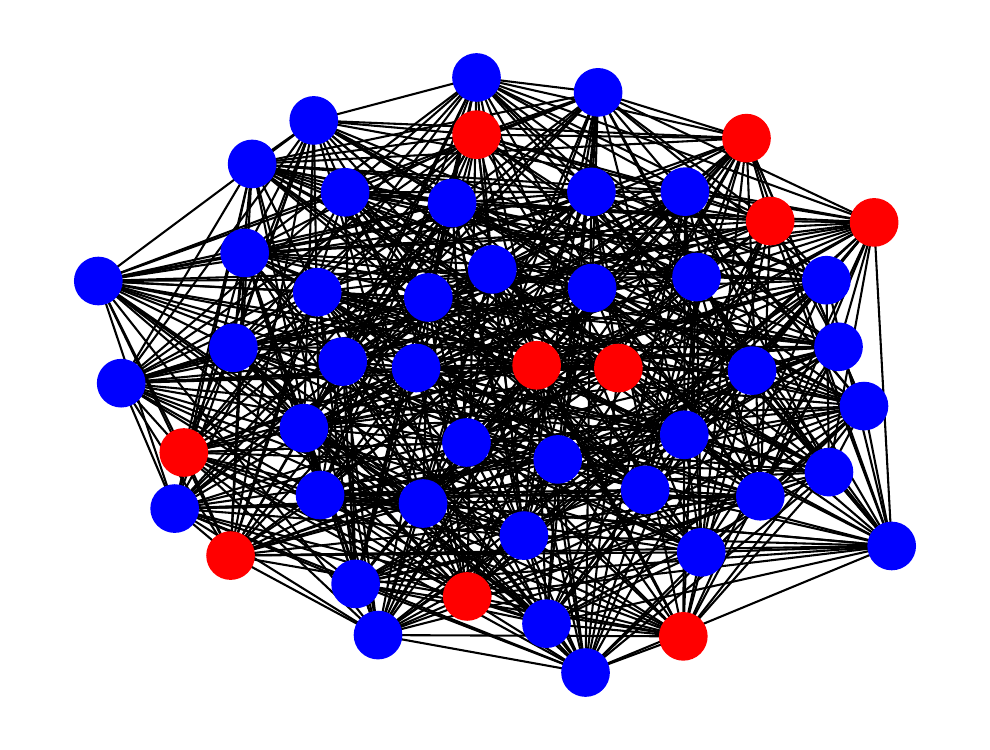}\label{fig_no_label_figure_regular_n50_25deg}} 
	\subfloat[Complete graph.]{\includegraphics[width=0.152 \textwidth]{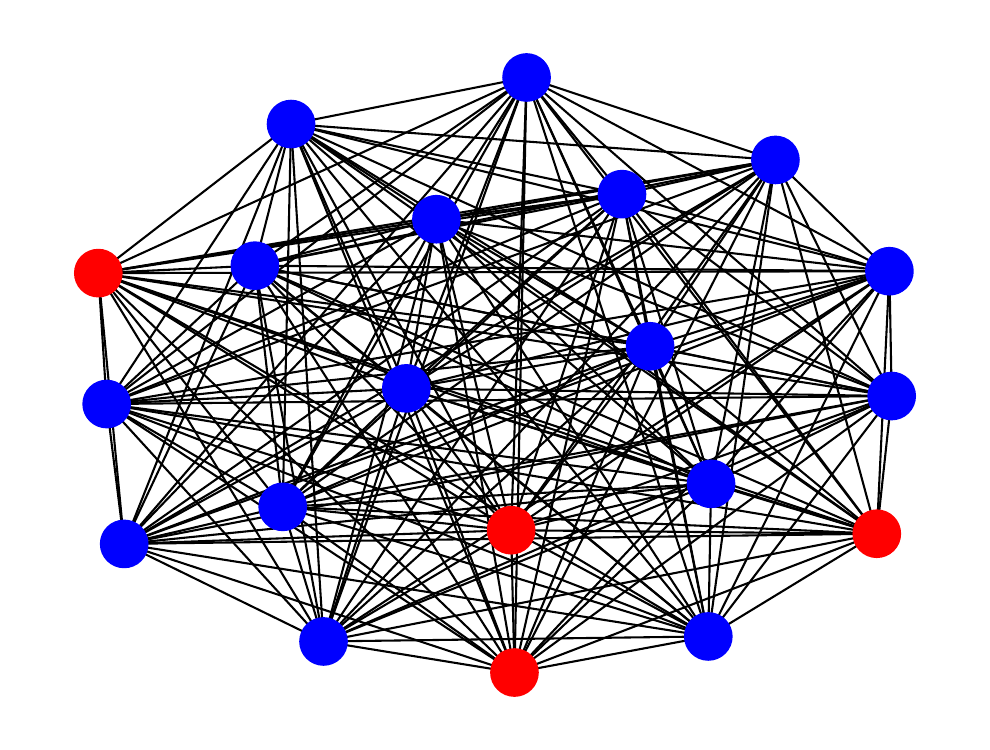}\label{fig_no_label_figure_adjacency_full_20}}
	\\
	\subfloat[Erdős–Rényi graph.]{\includegraphics[width=0.152 \textwidth]{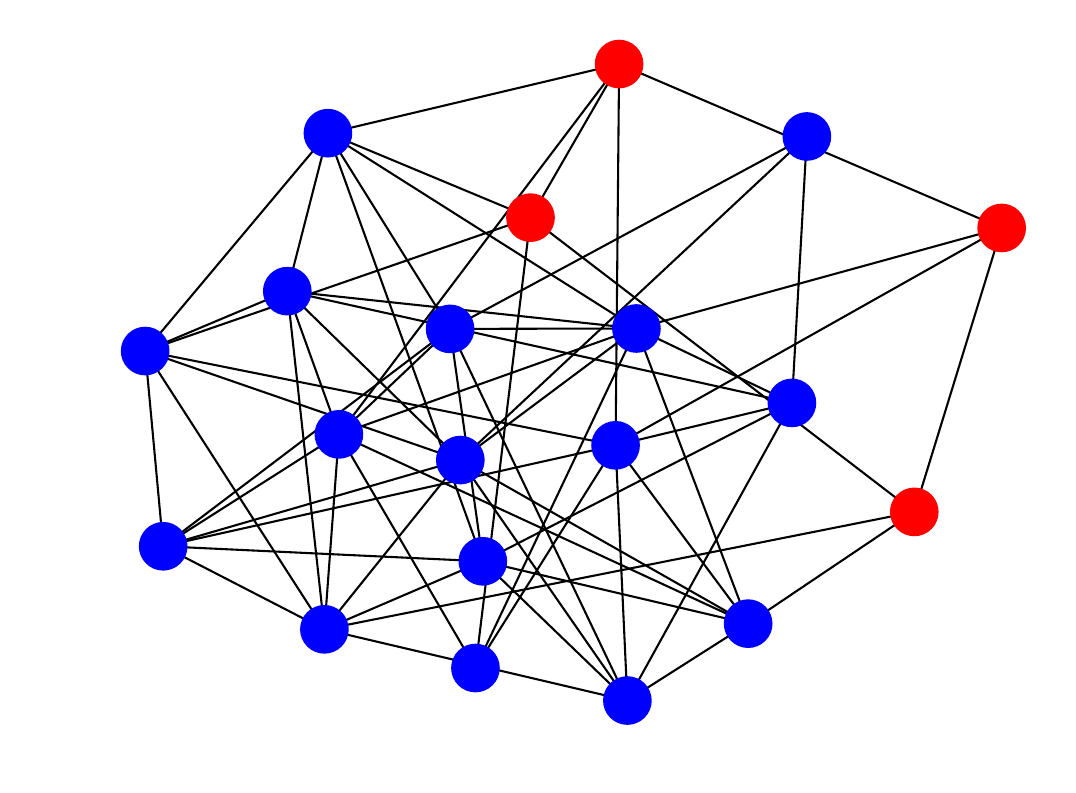}\label{fig_ER_graph}} 
	\subfloat[Small-world graph.]{\includegraphics[width=0.152 \textwidth]{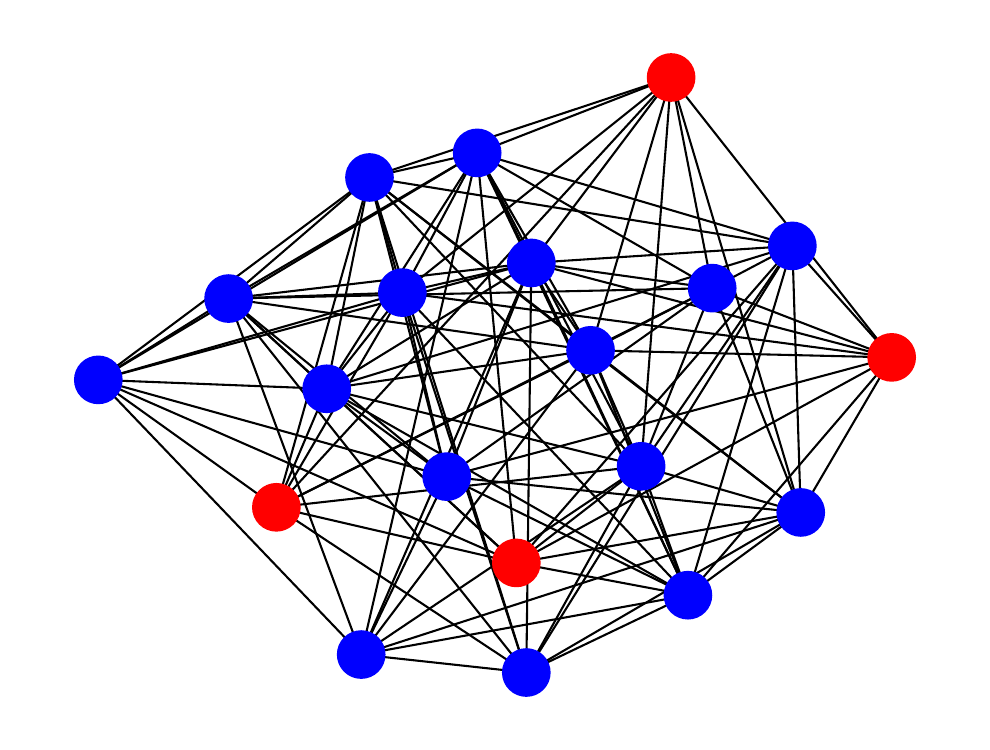}\label{fig_small_graph}} 
	\subfloat[Ring graph.]{\includegraphics[width=0.152 \textwidth]{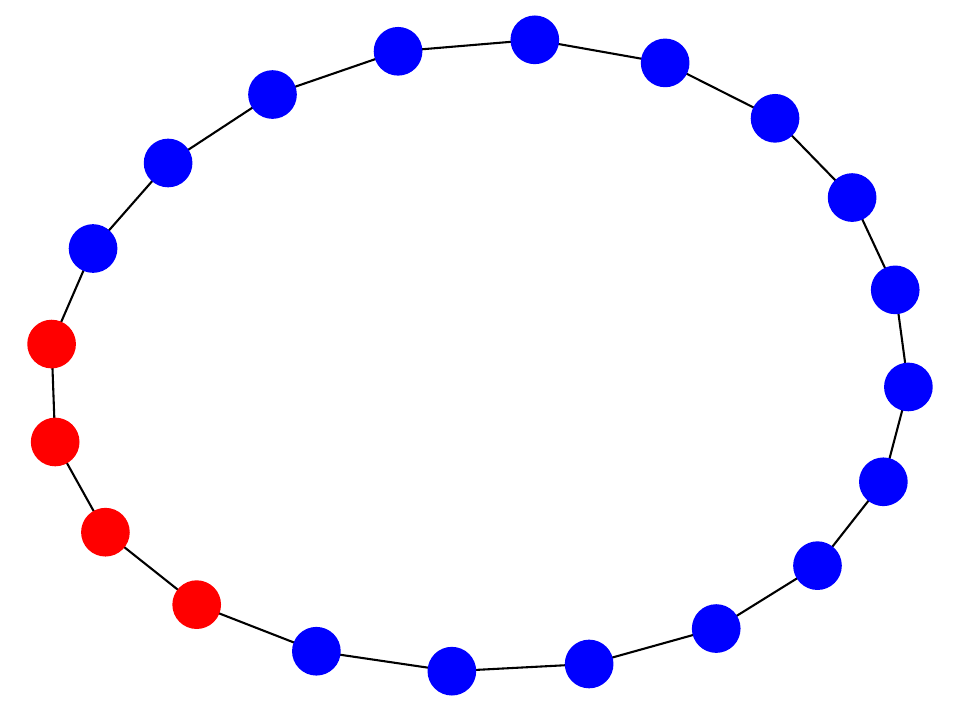}\label{fig_ring_graph}} 
	\subfloat[Graph with FEMB 0.16.]{\includegraphics[width=0.152 \textwidth]{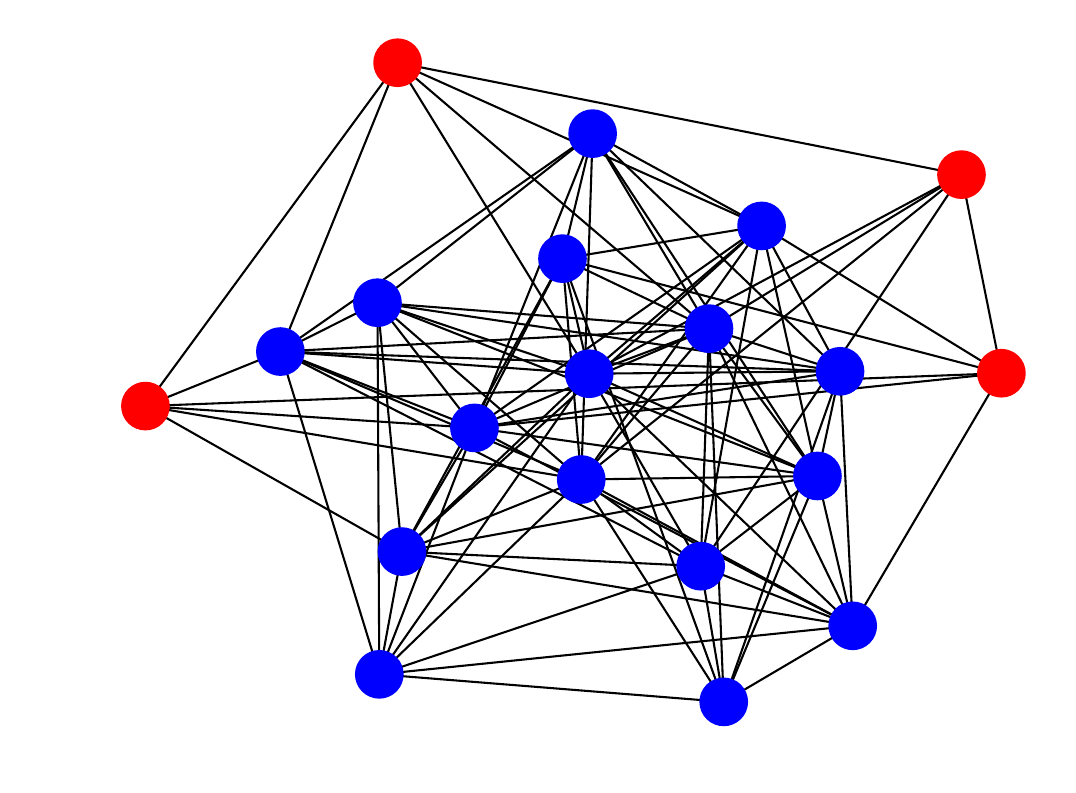}\label{fig_FEMB_016}} 
	\subfloat[Graph with FEMB 0.22.]{\includegraphics[width=0.152 \textwidth]{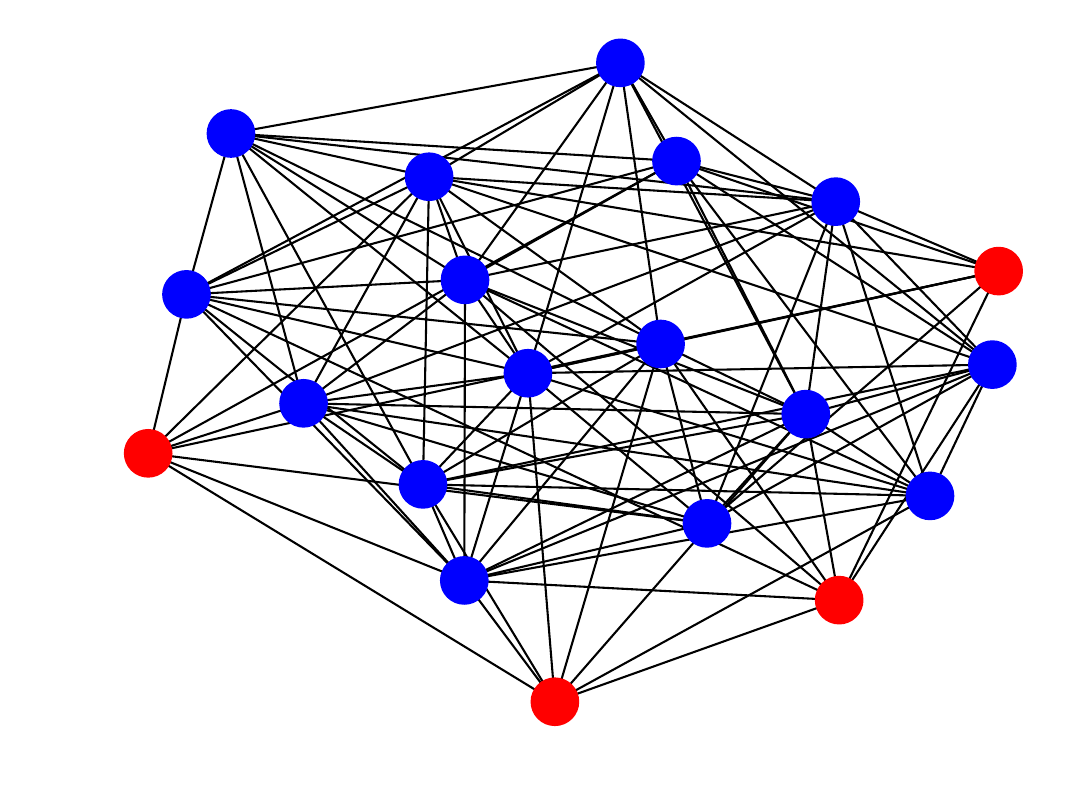}\label{fig_FEMB_022}}
	\subfloat[Graph with FEMB 0.32.]{\includegraphics[width=0.152 \textwidth]{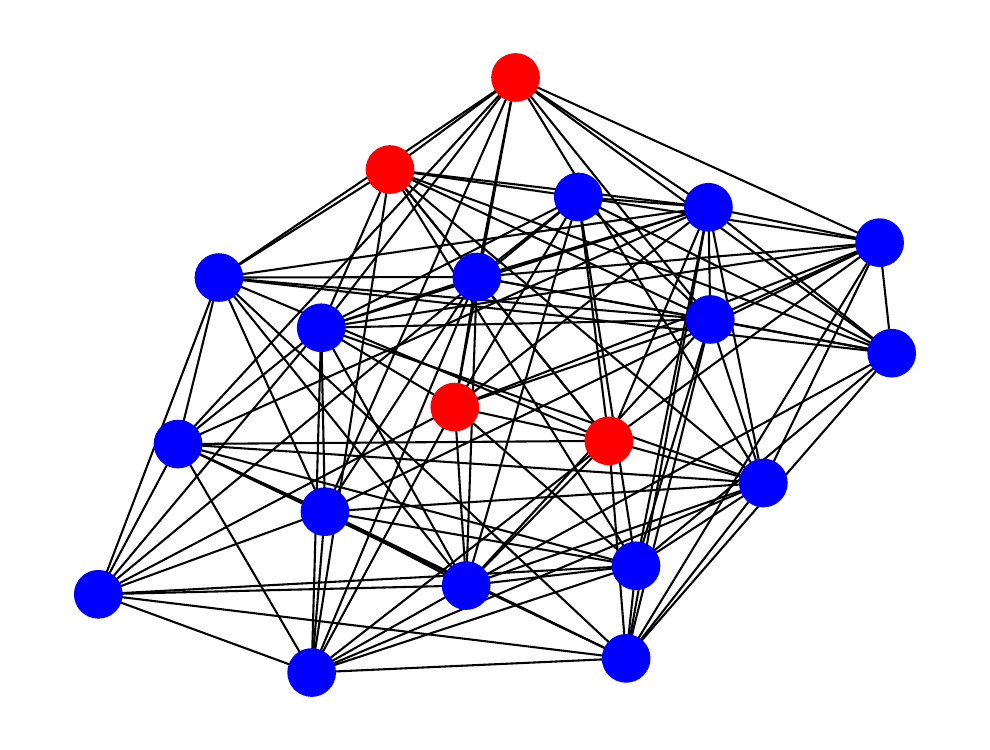}\label{fig_FEMB_032}} 
	\caption{Different communication graphs.}
	\label{fig_diff_graph}
\end{figure*}

\begin{figure*}[!t]
	\centering
	\includegraphics[scale = 0.46]{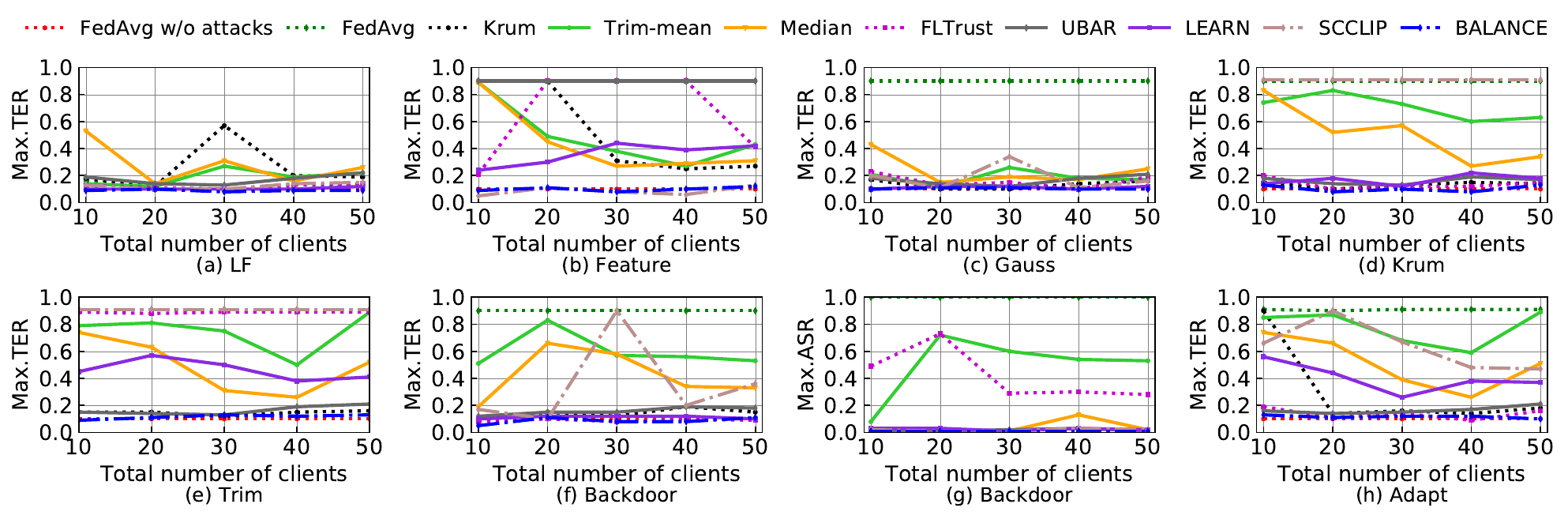}
	\caption{Impact of total number of clients.
	}
	\label{number_of_clients_max_error}
\end{figure*}

\section{Datasets and Poisoning Attacks} \label{sec:setting_app}

\subsection{Datasets}
\label{sec:datasets_app}

\myparatight{Synthetic}
To create synthetic data, we model the dependent variable as \(y = \langle \bm{x}, \bm{w}^* \rangle + \epsilon\), with \(\bm{x}\) as a feature vector, \(\bm{w}^*\) as the true parameter (100-dimensional), and \(\epsilon\) as noise. We generate \(\bm{x}\) and \(\epsilon\) from a standard normal distribution \(N(0, 1)\), and \(\bm{w}^*\) from \(N(0, 25)\). Our dataset comprises 10,000 instances, split into 8,000 for training and 2,000 for testing.

\myparatight{MNIST~\cite{lecun2010mnist}}
This dataset has 10 classes, which contains 60,000 images for training and 10,000 images for testing.

\myparatight{Fashion-MNIST~\cite{xiao2017online}}
Each image in Fashion-MNIST belongs to one of the 10 categories. The training data contains 60,000 images, and the testing set contains 10,000 images.

\myparatight{Human Activity Recognition (HAR)~\cite{anguita2013public}} 
The HAR dataset aims to classify six human activities, collected from 30 smartphone users, totaling 10,299 instances with 561 features each. Follow~\cite{cao2020fltrust}, we randomly use 75\% of each user's data for training and 25\% for testing.

\myparatight{Large-scale CelebFaces Attributes (CelebA)~\cite{liu2015faceattributes}} 
 CelebA is a large-scale image dataset that identifies celebrity face attributes. This data contains 200,288 images, which includes 177,480 images for training and 22,808 images for testing. Each image has 40 annotations of binary attributes. Following~\cite{caldas2018leaf}, we consider the binary classification task for the CelebA dataset, which aims to predict whether the person in the image is smiling or not.

\subsection{Poisoning Attack Schemes}
\label{sec:attack_app}

\myparatight{Label flipping (LF) attack~\cite{tolpegin2020data}} 
In the synthetic data, the attacker modifies malicious clients' local training data by adding a bias of 5 to the dependent variable \(y\). For MNIST, Fashion-MNIST, and HAR datasets, training labels on malicious clients are changed from class 3 to class 5, following~\cite{tolpegin2020data}. In CelebA, labels are reversed on malicious clients, switching 0 to 1 and vice versa.

\myparatight{Feature attack}
The attacker modifies the features of local training examples on malicious clients. Each feature of such examples is replaced with a value drawn from a Gaussian distribution with a mean of 0 and a variance of 1,000.

\myparatight{Gaussian (Gauss) attack~\cite{blanchard2017machine}} 
Malicious clients send Gaussian vectors, randomly drawn from a normal distribution with a mean of 0 and a variance of 200, to their neighbors.

\myparatight{Krum attack~\cite{fang2020local}} 
The attacker carefully crafts the local models on malicious clients in a way that causes the Krum rule to output the model chosen by the attacker.

\myparatight{Trim attack~\cite{fang2020local}} 
The attacker in Trim attack manipulates the local models on malicious clients such that the aggregated local model after attack deviates significantly from the before-attack aggregated one.

\myparatight{Backdoor attack~\cite{bagdasaryan2020backdoor,gu2017badnets}} 
Malicious clients replicate their training data, adding a backdoor trigger to each copy and assigning them a target label chosen for the attack. They train their local models using this augmented data and share the scaled models with neighbors. The scaling factor equals the total number of clients.
We use the triggers suggested in~\cite{cao2020fltrust} for the MNIST, Fashion-MNIST, and HAR datasets.
For the CelebA dataset, we set the first binary feature to 1.

\myparatight{Adaptive (Adapt) attack~\cite{shejwalkar2021manipulating}} 
We consider the adaptive attack proposed in~\cite{shejwalkar2021manipulating}.
We consider the worst-case attack setting, where the attacker is aware of the aggregation rule used by each client, i.e., \alg in our paper, and the local models of benign clients.

\subsection{Consensus Error}
\label{sec:more_metrics}

To assess disagreement among benign clients during poisoning attacks, we use the consensus error metric~\cite{he2022byzantine,lian2017can,kong2021consensus}, which is computed as \(\frac{1}{|\mathcal{B}|} \sum_{i \in \mathcal{B}} \| \bm{w}_i^{T} - \bar{\bm{w}}^T \|^2\). Here, $\mathcal{B}$ is the set of benign clients, \(\bar{\bm{w}}^T\) is the average of their final local models, and \(\bm{w}_i^{T}\) is client \(i\)'s final model after \(T\) training rounds in the DFL system.

\begin{table}[H]
	\centering
	\scriptsize
	\setlength{\doublerulesep}{3\arrayrulewidth}
		\caption{Maximum mean squared errors (Max.MSEs) of different DFL methods on synthetic dataset.}
	\label{tab_all_datasets_app_Synthetic}%
	\begin{tabular}{|c|c|c|c|c|c|c|c|}
		\hline
		Method & No & LF & Feature & Gauss & Krum & Trim  & Adapt \\
		\hline
		\hline
		FedAvg  & 0.36  & 0.39    & >100    & >100    &  72.18    &  58.50    &  90.29  \\
		Krum    & 1.11    & 1.23     & >100   & 1.19&   1.18  &   1.18   &  1.19    \\
		Trim-mean   & 0.38    &  0.39   & 3.45   &  0.40 &   4.17  &   5.41     &  5.41  \\
		Median   &  0.39   &  0.40     &  2.37  & 0.42   & 1.22   &  3.93   & 3.93  \\
		FLTrust  &0.41     &   0.46    &  >100    & 22.87     &  10.12     &  8.00     &  0.42  \\
		UBAR   & 0.40  &   0.40     & >100    & 0.40     &   0.40     & 0.40    &  0.40   \\
		LEARN  &   0.42  &  0.42    &  0.42  & 0.64    &  5.36   & 17.78     & 1.60  \\
		SCCLIP    &  0.36  & 0.39    & >100     & 0.42  & 5.63  &   5.12  &  4.83  \\
		\rowcolor{greyL}
		\alg     &  0.36   &  0.36     & 0.36  & 0.36  &   0.36     &   0.36    &   0.36  \\
		\hline
	\end{tabular}%
	\label{tab_all_datasets_app_Synthetic}%
\end{table}%

\begin{table}[H]
	\centering
	\scriptsize
	\setlength{\doublerulesep}{3\arrayrulewidth}
	\caption{Consensus error of different DFL methods.}
	\label{tab_avg_error}%
	{
		\begin{tabular}{|c|c|c|c|c|c|c|c|c|}
			\hline
			Method & No & LF & Feature & Gauss & Krum & Trim & Backdoor & Adapt  \\
			\hline
			\hline
			FedAvg & 0.01 & 0.01 & >100 & >100 & >100 & 0.01 & >100 & >100  \\
			Krum & 0.01 & 0.01 & >100 & 0.01 & 0.01 & 0.01 & 0.01 & 0.01  \\
			Trim-mean & 0.01 & 0.01 & 0.01 & 0.01 & 0.01 & 0.01 & 0.02 & 0.01 \\
			Median & 0.01 & 0.01 & 0.01 & 0.01 & 0.01 & 0.01 & 0.56 & 0.01 \\
			FLTrust & 0.01 & 0.01 & 45.32 & 2.13 & 0.01 & 1.71 & 0.01 & 0.01\\
			UBAR & 0.01 & 0.01 & >100 & 0.01 & 0.01 & 0.01 & 0.01 & 0.01\\
			LEARN & 0.01 & 0.01 & 0.01 & 0.01 & 0.01 & 0.01 & 0.01 & 0.01  \\
			SCCLIP & 0.01 & 0.01 & 0.02 & 0.01 & 0.01 & 0.01 & >100 & 0.01 \\
			\rowcolor{greyL}
			\alg & 0.01 & 0.01 & 0.01 & 0.01 & 0.01 & 0.01 & 0.01 & 0.01  \\
			\hline
		\end{tabular}%
	}
	\vspace{-.15in}
\end{table}%

\begin{table}[H]
	\centering
	\scriptsize
	\setlength{\doublerulesep}{3\arrayrulewidth}
	\centering
	\caption{Results of different DFL methods, where each client aggregates its model as $\bm{w}_i^{t+1} = \text{AGG} \{ \bm{w}_j^{t+\frac{1}{2}}, j \in \widehat{\mathcal{N}_i} \}$, $\widehat{\mathcal{N}_i} = \mathcal{N}_i \cup \{ i\}$, $\mathcal{N}_i$ is the set of neighbors of client $i$ (not including client $i$ itself).}
	\label{tab_aggre_all}
	\begin{tabular}{|c|c|c|c|c|c|c|c|c|}
		\hline
		Method & No & LF & Feature & Gauss & Krum & Trim & Backdoor & Adapt\\
		\hline
		\hline
		FedAvg & 0.09 & 0.09 & 0.90 & 0.90 & 0.91 & 0.91 & 0.90 / 1.00 & 0.90 \\
		Krum & 0.18 & 0.18 & 0.13 & 0.18 & 0.22 & 0.18 & 0.18 / 0.01 & 0.18 \\
		Trim-mean & 0.17 & 0.49 & 0.17 & 0.17 & 0.89 & 0.89 & 0.90 / 1.00 & 0.89 \\
		Median & 0.19 & 0.38 & 0.23 & 0.31 & 0.83 & 0.87 & 0.78 / 0.01 & 0.88 \\
		FLTrust & 0.11 & 0.11 & 0.11 & 0.11 & 0.11 & 0.90 & 0.11 / 0.90 & 0.12\\
		UBAR & 0.25 & 0.25 & 0.26 & 0.25 & 0.27 & 0.27 & 0.25 / 0.01 & 0.23  \\
		LEARN & 0.10 & 0.15 & 0.11 & 0.10 & 0.35 & 0.55 & 0.10 / 0.19 & 0.81 \\
		SCCLIP & 0.10 & 0.10 & 0.10 & 0.10 & 0.90 & 0.91 & 0.46 / 0.04 & 0.91\\
		\rowcolor{greyL}
		\alg & 0.09 & 0.09 & 0.09 & 0.10 & 0.09 & 0.09 & 0.09 / 0.01 & 0.10 \\
		\hline
	\end{tabular}%
\end{table}%

\begin{table}[H]
    \centering
    \scriptsize
    \setlength{\doublerulesep}{3\arrayrulewidth}
    \caption{Results of different DFL methods in Case I.}
    \label{tab_diff_alpha}%
    {
\begin{tabular}{|c|c|c|c|c|c|c|c|c|}
    \hline
    Method & No & LF & Feature & Gauss & Krum & Trim & Backdoor & Adapt \\
    \hline
    \hline
    FedAvg & 0.16 & 0.17 & 0.90 & 0.90 & 0.91 & 0.91 & 0.90 / 1.00 & 0.90 \\
    Krum & 0.23 & 0.24 & 0.90 & 0.26 & 0.26 & 0.28 & 0.26 / 0.01 & 0.26\\
    Trim-mean & 0.20 & 0.21 & 0.15 & 0.34 & 0.44 & 0.74 & 0.23 / 0.03 & 0.48 \\
    Median & 0.22 & 0.23 & 0.23 & 0.22 & 0.56 & 0.60 & 0.24 / 0.02 & 0.33 \\
    FLTrust & 0.16 & 0.17 & 0.90 & 0.16 & 0.17 & 0.89 & 0.22 / 0.16 & 0.91 \\
    UBAR & 0.20 & 0.20 & 0.90 & 0.20 & 0.21 & 0.22 & 0.23 / 0.03 & 0.22\\
    LEARN & 0.16 & 0.21 & 0.25 & 0.17 & 0.30 & 0.64 & 0.16 / 0.05 & 0.38  \\
    SCCLIP & 0.16 & 0.17 & 0.17 & 0.17 & 0.87 & 0.90 & 0.16 / 0.01 & 0.90 \\
    \rowcolor{greyL}
    \alg & 0.16 & 0.16 & 0.17 & 0.17 & 0.17 & 0.17 & 0.16 / 0.01 & 0.18 \\
    \hline
    \end{tabular}%
    }
	\vspace{-.15in}
\end{table}%

\begin{table}[H]
	\centering
    \scriptsize
    \setlength{\doublerulesep}{3\arrayrulewidth}
        \caption{Results of different DFL methods in Case II.}
    \label{tab_diff_alpha_app}
\begin{tabular}{|c|c|c|c|c|c|c|c|c|}
        \hline
        Method & No & LF & Feature & Gauss & Krum & Trim & Backdoor & Adapt \\
        \hline
        \hline
        FedAvg & 0.19 & 0.17 & 0.90 & 0.90 & 0.90 & 0.90 & 0.90 / 1.00 & 0.90 \\
        Krum & 0.21 & 0.22 & 0.90 & 0.27 & 0.25 & 0.26 & 0.26 / 0.02 & 0.27 \\
        Trim-mean & 0.43 & 0.43 & 0.57 & 0.43 & 0.89 & 0.89 & 0.90 / 1.00 & 0.89 \\
        Median & 0.39 & 0.40 & 0.46 & 0.44 & 0.79 & 0.77 & 0.53 / 0.05 & 0.77 \\
        FLTrust & 0.19 & 0.19 & 0.90 & 0.24 & 0.21 & 0.91 & 0.19 / 0.86 & 0.90  \\
        UBAR & 0.24 & 0.24 & 0.90 & 0.31 & 0.30 & 0.24 & 0.46 / 0.03 & 0.29\\
        LEARN & 0.32 & 0.35 & 0.32 & 0.32 & 0.75 & 0.89 & 0.33 / 0.04 & 0.87 \\
        SCCLIP & 0.19 & 0.19 & 0.20 & 0.30 & 0.90 & 0.90 & 0.39 / 0.02 & 0.90  \\
        \rowcolor{greyL}
        \alg & 0.19 & 0.19 & 0.20 & 0.20 & 0.21 & 0.22 & 0.20 / 0.01 & 0.22\\
        \hline
    \end{tabular}%
\end{table}%

\begin{table}[H]
 	\scriptsize
    \setlength{\doublerulesep}{3\arrayrulewidth}
	\centering
        \caption{Results of different DFL methods in Case III and Case IV.}
        \label{tab_diff_agg}%
	\begin{tabular}{|c|c|c|c|c|c|c|c|c|}
		\hline
		Method & No & LF & Feature & Gauss & Krum & Trim & Backdoor & Adapt \\
		\hline
      \hline
		Case III  &   0.14      & 0.14      &  0.90     & 0.24     &  0.91      & 0.89     &          0.16 / 0.01   & 0.90 \\
		Case IV  &   0.17   &  0.19  &  0.90    & 0.09 &  0.75     &  0.90     &  0.23 / 0.05   &   0.80   \\
		\hline
	\end{tabular}%
\end{table}%

\begin{table}[H]
    \centering
    \scriptsize
    \setlength{\doublerulesep}{3\arrayrulewidth}
        \caption{Results of different DFL methods, where clients use different initial local models.}
    \label{tab_diff_init}
    \begin{tabular}{|c|c|c|c|c|c|c|c|c|}
        \hline
        Method & No & LF & Feature & Gauss & Krum & Trim & Backdoor & Adapt  \\
        \hline
        \hline
        FedAvg & 0.10 & 0.11 & 0.90 & 0.90 & 0.90 & 0.90 & 0.90 / 1.00 & 0.90  \\
        Krum & 0.10 & 0.12 & 0.90 & 0.10 & 0.12 & 0.12 & 0.15 / 0.01 & 0.12  \\
        Trim-mean & 0.11 & 0.14 & 0.44 & 0.11 & 0.84 & 0.70 & 0.91 / 0.01 & 0.59 \\
        Median & 0.13 & 0.17 & 0.56 & 0.16 & 0.58 & 0.64 & 0.19 / 0.01 & 0.87  \\
        FLTrust & 0.11 & 0.11 & 0.11 & 0.13 & 0.11 & 0.89 & 0.10 / 0.47 & 0.11 \\
        UBAR & 0.14 & 0.15 & 0.91 & 0.14 & 0.14 & 0.14 & 0.14 / 0.01 & 0.16\\
        LEARN & 0.13 & 0.14 & 0.14 & 0.14 & 0.25 & 0.36 & 0.15 / 0.06 & 0.31 \\
        SCCLIP & 0.10 & 0.10 & 0.10 & 0.11 & 0.91 & 0.91 & 0.13 / 0.01 & 0.91 \\
        \rowcolor{greyL}
        \alg & 0.10 & 0.10 & 0.10 & 0.10 & 0.10 & 0.10 & 0.10 / 0.01 & 0.11 \\
        \hline
    \end{tabular}%
\end{table}%

\begin{table}[H]
	\centering
	\scriptsize
	\setlength{\doublerulesep}{3\arrayrulewidth}
        \caption{Results of different DFL methods on different communication graphs.}
	\label{tab_diff_graph_app}
     \subfloat[Complete graph.]
     {
\begin{tabular}{|c|c|c|c|c|c|c|c|c|}
        \hline
        Method & No & LF & Feature & Gauss & Krum & Trim & Backdoor & Adapt \\
        \hline
        \hline
        FedAvg & 0.10 & 0.10 & 0.90 & 0.90 & 0.91 & 0.91 & 0.90 / 1.00 & 0.90 \\
        Krum & 0.14 & 0.14 & 0.90 & 0.15 & 0.15 & 0.15 & 0.14 / 0.01 & 0.15 \\
        Trim-mean & 0.13 & 0.13 & 0.27 & 0.13 & 0.88 & 0.88 & 0.56 / 0.58 & 0.88 \\
        Median & 0.13 & 0.16 & 0.46 & 0.22 & 0.89 & 0.88 & 0.20 / 0.01 & 0.89  \\
        FLTrust & 0.11 & 0.11 & 0.11 & 0.13 & 0.11 & 0.91 & 0.10 / 0.66 & 0.11 \\
        UBAR & 0.16 & 0.15 & 0.91 & 0.17 & 0.17 & 0.17 & 0.19 / 0.01 & 0.17  \\
        LEARN & 0.13 & 0.27 & 0.31 & 0.14 & 0.33 & 0.49 & 0.11 / 0.02 & 0.44\\
        SCCLIP & 0.10 & 0.10 & 0.17 & 0.11 & 0.91 & 0.91 & 0.11 / 0.01 & 0.91 \\
        \rowcolor{greyL}
        \alg & 0.10 & 0.10 & 0.10 & 0.10 & 0.10 & 0.10 & 0.10 / 0.01 & 0.10  \\
        \hline
    \end{tabular}%
     }
     \\
        \subfloat[Erdős–Rényi graph.]
        {
\begin{tabular}{|c|c|c|c|c|c|c|c|c|}
			\hline
			Method & No & LF & Feature & Gauss & Krum & Trim & Backdoor & Adapt \\
			\hline
			\hline
                FedAvg & 0.10 & 0.10 & 0.90 & 0.91 & 0.90 & 0.87 & 0.90 / 1.00 & 0.90 \\
                Krum & 0.17 & 0.17 & 0.90 & 0.17 & 0.17 & 0.17 & 0.17 / 0.01 & 0.90  \\
                Trim-mean & 0.14 & 0.16 & 0.27 & 0.14 & 0.32 & 0.42 & 0.14 / 0.01 & 0.45 \\
                Median & 0.17 & 0.22 & 0.43 & 0.24 & 0.30 & 0.65 & 0.17 / 0.01 & 0.40  \\
                FLTrust & 0.10 & 0.10 & 0.11 & 0.11 & 0.09 & 0.90 & 0.10 / 0.04 & 0.10 \\
                UBAR & 0.23 & 0.23 & 0.90 & 0.23 & 0.23 & 0.23 & 0.24 / 0.02 & 0.23 \\
                LEARN & 0.10 & 0.11 & 0.10 & 0.12 & 0.19 & 0.27 & 0.10 / 0.01 & 0.12 \\
                SCCLIP & 0.10 & 0.10 & 0.09 & 0.11 & 0.85 & 0.89 & 0.10 / 0.01 & 0.59 \\
			\rowcolor{greyL}
			\alg & 0.10 & 0.10 & 0.10 & 0.10 & 0.10 & 0.10 & 0.10 / 0.01 & 0.10\\
			\hline
		\end{tabular}
	}
  \\
 	\subfloat[Small-world graph.]{
	\begin{tabular}{|c|c|c|c|c|c|c|c|c|c|}
			\hline
			Method & No & LF & Feature & Gauss & Krum & Trim & Backdoor & Adapt   \\
			\hline
			\hline
                FedAvg & 0.10 & 0.10 & 0.90 & 0.90 & 0.91 & 0.91 & 0.90 / 1.00 & 0.90  \\
                Krum & 0.17 & 0.17 & 0.90 & 0.17 & 0.17 & 0.17 & 0.19 / 0.01 & 0.17  \\
                Trim-mean & 0.15 & 0.43 & 0.63 & 0.17 & 0.71 & 0.87 & 0.32 / 0.01 & 0.80 \\
                Median & 0.14 & 0.15 & 0.73 & 0.19 & 0.87 & 0.87 & 0.19 / 0.01 & 0.85  \\
                FLTrust & 0.10 & 0.10 & 0.13 & 0.13 & 0.10 & 0.91 & 0.10 / 0.14 & 0.11 \\
                UBAR & 0.14 & 0.14 & 0.90 & 0.14 & 0.14 & 0.14 & 0.18 / 0.01 & 0.14  \\
                LEARN & 0.10 & 0.11 & 0.16 & 0.10 & 0.19 & 0.33 & 0.10 / 0.01 & 0.29  \\
                SCCLIP & 0.10 & 0.10 & 0.10 & 0.10 & 0.90 & 0.91 & 0.10 / 0.01 & 0.49  \\
			\rowcolor{greyL}
			\alg & 0.10 & 0.10 & 0.11 & 0.11 & 0.11 & 0.11 & 0.10 / 0.01 & 0.11 \\
			\hline
		\end{tabular}
	}
 \\
 	\subfloat[Ring graph.]{
	\begin{tabular}{|c|c|c|c|c|c|c|c|c|}
			\hline
			Method & No & LF & Feature & Gauss & Krum & Trim & Backdoor & Adapt   \\
			\hline
			\hline
                FedAvg & 0.10 & 0.11 & 0.90 & 0.90 & 0.90 & 0.91 & 0.90 / 1.00 & 0.90 \\
                Krum & 0.17 & 0.17 & 0.90 & 0.90 & 0.90 & 0.90 & 0.90 / 1.00 & 0.90 \\
                Trim-mean & 0.19 & 0.21 & 0.90 & 0.90 & 0.90 & 0.90 & 0.90 / 1.00 & 0.90 \\
                Median & 0.10 & 0.11 & 0.90 & 0.90 & 0.90 & 0.90 & 0.90 / 1.00 & 0.90 \\
                FLTrust & 0.11 & 0.11 & 0.11 & 0.36 & 0.11 & 0.90 & 0.11 / 0.49 & 0.11  \\
                UBAR & 0.13 & 0.13 & 0.90 & 0.13 & 0.13 & 0.13 & 0.27 / 0.01 & 0.90  \\
                LEARN & 0.12 & 0.12 & 0.90 & 0.90 & 0.90 & 0.90 & 0.90 / 1.00 & 0.90 \\
                SCCLIP & 0.11 & 0.11 & 0.13 & 0.11 & 0.83 & 0.90 & 0.11 / 0.01 & 0.49\\
			\rowcolor{greyL}
			\alg & 0.10 & 0.10 & 0.11 & 0.10 & 0.10 & 0.10 & 0.11 / 0.01 & 0.12 \\
			\hline
		\end{tabular}
	}
\end{table}%

\newpage

\begin{table}[H]
	\centering
	\scriptsize
	\setlength{\doublerulesep}{3\arrayrulewidth}
        \caption{Results of different DFL methods on graphs with different FEMBs.}
	\label{tab_diff_FEMB_results_app}
    \subfloat[Graph with FEMB 0.16.]{
\begin{tabular}{|c|c|c|c|c|c|c|c|c|}
        \hline
        Method & No & LF & Feature & Gauss & Krum & Trim & Backdoor & Adapt \\
        \hline
        \hline
        FedAvg & 0.10 & 0.10 & 0.90 & 0.90 & 0.90 & 0.91 & 0.90 / 1.00 & 0.90 \\
        Krum & 0.17 & 0.17 & 0.90 & 0.17 & 0.18 & 0.17 & 0.17 / 0.01 & 0.19 \\
        Trim-mean & 0.13 & 0.13 & 0.13 & 0.13 & 0.31 & 0.39 & 0.13 / 0.01 & 0.33\\
        Median & 0.14 & 0.15 & 0.20 & 0.15 & 0.43 & 0.41 & 0.14 / 0.01 & 0.43 \\
        FLTrust & 0.10 & 0.10 & 0.11 & 0.11 & 0.11 & 0.91 & 0.10 / 0.01 & 0.10 \\
        UBAR & 0.17 & 0.17 & 0.90 & 0.18 & 0.18 & 0.18 & 0.17 / 0.01 & 0.18 \\
        LEARN & 0.10 & 0.10 & 0.10 & 0.10 & 0.14 & 0.28 & 0.10 / 0.01 & 0.10\\
        SCCLIP & 0.10 & 0.10 & 0.10 & 0.11 & 0.81 & 0.91 & 0.10 / 0.01 & 0.49   \\
        \rowcolor{greyL}
        \alg & 0.10 & 0.10 & 0.10 & 0.10 & 0.10 & 0.10 & 0.10 / 0.01 & 0.10 \\
        \hline
    \end{tabular}%
        }
\\
	\subfloat[Graph with FEMB 0.22.]{
\begin{tabular}{|c|c|c|c|c|c|c|c|c|}
			\hline
			Method & No & LF & Feature & Gauss & Krum & Trim & Backdoor & Adapt  \\
			\hline
			\hline
			FedAvg & 0.10 & 0.10 & 0.90 & 0.90 & 0.90 & 0.91 & 0.90 / 1.00 & 0.90 \\
                Krum & 0.17 & 0.17 & 0.90 & 0.17 & 0.17 & 0.17 & 0.17 / 0.01 & 0.17  \\
                Trim-mean & 0.13 & 0.13 & 0.13 & 0.13 & 0.19 & 0.55 & 0.17 / 0.01 & 0.51\\
                Median & 0.14 & 0.14 & 0.14 & 0.17 & 0.54 & 0.62 & 0.22 / 0.01 & 0.59\\
                FLTrust & 0.10 & 0.10 & 0.11 & 0.11 & 0.10 & 0.89 & 0.10 / 0.04 & 0.10\\
                UBAR & 0.16 & 0.16 & 0.90 & 0.17 & 0.18 & 0.18 & 0.18 / 0.01 & 0.16\\
                LEARN & 0.10 & 0.13 & 0.12 & 0.10 & 0.12 & 0.43 & 0.10 / 0.01 & 0.19  \\
                SCCLIP & 0.10 & 0.10 & 0.10 & 0.10 & 0.89 & 0.90 & 0.10 / 0.01 & 0.43\\
			\rowcolor{greyL}
			\alg & 0.10 & 0.10 & 0.11 & 0.10 & 0.10 & 0.11 & 0.10 / 0.01 & 0.11\\
			\hline
		\end{tabular}
	}
\\
	\subfloat[Graph with FEMB 0.32.]{
\begin{tabular}{|c|c|c|c|c|c|c|c|c|}
			\hline
			Method & No & LF & Feature & Gauss & Krum & Trim & Backdoor & Adapt   \\
			\hline
			\hline
                FedAvg & 0.10 & 0.10 & 0.90 & 0.90 & 0.90 & 0.91 & 0.90 / 1.00 & 0.90  \\
                Krum & 0.17 & 0.17 & 0.90 & 0.17 & 0.19 & 0.17 & 0.17 / 0.01 & 0.17 \\
                Trim-mean & 0.12 & 0.13 & 0.19 & 0.13 & 0.86 & 0.87 & 0.59 / 0.60 & 0.88\\
                Median & 0.14 & 0.14 & 0.14 & 0.19 & 0.89 & 0.88 & 0.86 / 0.01 & 0.89  \\
                FLTrust & 0.10 & 0.10 & 0.12 & 0.12 & 0.10 & 0.91 & 0.10 / 0.01 & 0.11 \\
                UBAR & 0.16 & 0.17 & 0.91 & 0.17 & 0.19 & 0.19 & 0.21 / 0.01 & 0.28 \\
                LEARN & 0.13 & 0.13 & 0.25 & 0.13 & 0.19 & 0.43 & 0.13 / 0.01 & 0.39   \\
                SCCLIP & 0.10 & 0.10 & 0.10 & 0.10 & 0.90 & 0.91 & 0.10 / 0.01 & 0.59  \\
			\rowcolor{greyL}
			\alg & 0.10 & 0.10 & 0.10 & 0.10 & 0.10 & 0.10 & 0.10 / 0.01 & 0.10\\
			\hline
		\end{tabular}
	}
\end{table}%

\end{document}